\def \doubleColumn {}
\ifx \doubleColumn \undefined % Single-Column
\documentclass[journal,draftcls,onecolumn,12pt,twoside]{IEEEtran}
\else
\documentclass[journal,10pt,twoside]{IEEEtran}
\fi

\usepackage{subfigure,cite,graphicx,amsmath,amssymb,mathrsfs,epsf}
\usepackage[usenames]{color}
\usepackage{multirow}
\usepackage{tabularx}
\usepackage{stfloats}
\usepackage{verbatim}
\usepackage{epsfig}
\usepackage{amsfonts}
\usepackage[normalem]{ulem}
\usepackage{amsthm}
\usepackage[table]{xcolor}
\usepackage{float}
\usepackage{color}
\usepackage[outdir=./]{epstopdf}
\IEEEoverridecommandlockouts
\newcommand{\Ic}{\mathcal{I}}
\normalsize

\begin{document}
\title{On the Effects of Subpacketization in Content-Centric Mobile Networks}

\author{ \IEEEauthorblockN{Adeel~Malik, Sung~Hoon~Lim, \emph{Member}, \emph{IEEE},\\ and Won-Yong Shin, \emph{Senior~Member}, \emph{IEEE}}
\thanks{Manuscript received December 10, 2017; revised April 8, 2018; accepted April 19, 2018. This work was supported by the Basic Science Research Program through the National Research Foundation of Korea (NRF) funded by the Ministry of Education (NRF-2017R1D1A1A09000835, NRF-2017R1C1B1004192). \em (Corresponding author: Won-Yong Shin.)}
\thanks{A. Malik and W.-Y. Shin are with the Department of Computer Science and Engineering, Dankook University, Yongin 16890, Republic of Korea (e-mail: adeel\_malik91@yahoo.com; wyshin@dankook.ac.kr).}
\thanks{S. H. Lim is with the Korea Institute of Ocean Science and Technology, Ansan 15627, Republic of Korea (e-mail: shlim@kiost.ac.kr).}
}

% The paper headers
\markboth{IEEE Journal on Selected Areas in Communications}%
{Manuscript}
% make the title area
\maketitle

\newtheorem{axiom}{Axiom}
\newtheorem{lemma}{Lemma}
\newtheorem{theorem}{Theorem}
\newtheorem{prop}{Proposition}
\newtheorem{observation}{Observation}
\newtheorem{definition}{Definition}
\newtheorem{remark}{Remark}%[section]
\newtheoremstyle{case}{}{}{}{}{}{:}{ }{}
\theoremstyle{case}
\newtheorem{case}{Case}

\begin{abstract}
A large-scale content-centric mobile ad hoc network employing {\em subpacketization} is studied in which each mobile node having finite-size cache moves according to the {\em reshuffling mobility model} and requests a content object from the library independently at random according to the Zipf popularity distribution. Instead of assuming that one content object is transferred in a single time slot, we consider a more challenging scenario where the size of each content object is considerably large and thus only a subpacket of a file can be delivered during one time slot, which is motivated by a {\em fast mobility} scenario. Under our mobility model, we consider a single-hop-based content delivery and characterize the fundamental trade-offs between throughput and delay. The order-optimal throughput--delay trade-off is analyzed by presenting the following two content reception strategies: the sequential reception for uncoded caching and the random reception for maximum distance separable (MDS)-coded caching. We also perform numerical evaluation to validate our analytical results. In particular, we conduct performance comparisons between the uncoded caching and the MDS-coded caching strategies by identifying the regimes in which the performance difference between the two caching strategies becomes prominent with respect to system parameters such as the Zipf exponent and the number of subpackets. In addition, we extend our study to the random walk mobility scenario and show that our main results are essentially the same as those in the reshuffling mobility model.

\end{abstract}

\IEEEpeerreviewmaketitle

\section{Introduction}~\label{section:1}
Wireless data caching plays an important role in maintaining the sustainability of future wireless networks by reducing the backhaul rate and the latency for retrieving content from networks without incurring any additional load on costly backhaul links\cite{R1-2,R1-3}. The core idea of caching is to bring content objects closer to the users by allowing the end terminals or helper nodes to cache a subset of popular content files locally.
\subsection{Prior Work}~\label{section:11}
The analysis of capacity scaling laws in large-scale wireless networks has attracted wide attention due to the dramatic growth of communication entities in today's networks. The pioneering work characterizing the capacity scaling law of static ad hoc networks having $n$ randomly distributed source--destination pairs in a unit network area was presented in~\cite{gupta}, in which the per-node throughput of $\Theta\left(\frac{1}{\sqrt{n \log n}}\right)$ was shown to be achievable using a nearest neighbor multihop transmission strategy. There have been further studies on multihop schemes in~\cite{R1-5,R1-6,R1-8}, where the per-node throughput scales far slower than $\Theta(1)$. In addition to the multihop schemes, there has been a steady push to improve the per-node throughput of wireless networks up to a constant scaling by using novel techniques such as networks with node mobility~\cite{algamal,grossglauser}, hierarchical cooperation~\cite{R1-9}, infrastructure support~\cite{R1-16,R1-17}, and directional antennas~\cite{R1-13,R1-15}.

In sharp contrast to the studies on ad hoc network modeling in which sources and destinations are given and fixed, investigating {\em content-centric ad hoc networks} would be quite challenging. As content objects are cached by numerous nodes over a network, finding the nearest content source of each request and scheduling between requests play a vital role in improving the overall network performance. The scaling behavior of content-centric ad hoc networks has received a lot of attention in the literature~\cite{alfano,as_law,jeon,R1-19}. In \cite{as_law,jeon}, throughput scaling laws were analyzed for {\em static} ad hoc networks using multihop communication, which yields a significant performance gain over the single-hop caching scenario~\cite{R1-3, R1-19}. More precisely, a decentralized and random cache allocation strategy along with a local multihop protocol was presented in~\cite{jeon}. A centralized and deterministic cache allocation strategy was presented in~\cite{as_law}, where replicas of each content object are statically determined based on the popularity of each content object. On the other hand, in {\em mobile} ad hoc networks, performance on the throughput and delay under a reshuffling mobility model was analyzed in~\cite{alfano}, where the position of each node is independently determined according to random walks (with an adjustable flight size) which is updated at the beginning of each time slot. It was shown in \cite{alfano} that increasing the mobility of nodes leads to worse performance. Performance on the throughput and delay under a correlated mobility model was investigated in~\cite{liu}, where nodes are partitioned into multiple clusters and the nodes belonging to the same cluster move  in a correlated fashion. It was shown in \cite{liu} how correlated mobility affects the network performance. In addition, the optimal throughput--delay trade-off in mobile hybrid networks was studied in~\cite{anh} when each request is served by mobile nodes or static base stations (or helper nodes) via multihop transmissions. It was shown in \cite{anh} that highly popular content objects are mainly served by mobile nodes while the rest of the content objects are served by static base stations.

Recently, a different caching framework, termed coded caching \cite{Maddah1,Maddah2, lim, d2d}, has received a lot of attention in content-centric wireless networks. To achieve the global caching gain, the content placement (caching) phase was optimized so that several different demands can be supported simultaneously with a single coded {\em multicast} transmission. Another promising topic is applications of maximum distance separable (MDS)-coded caching, in which MDS-coded subpackets of content objects are stored in local caches and the requested content objects are retrieved using {\em unicast} transmission. It has been shown in~\cite{mds1,mds2, femto} that with some careful placement of MDS-encoded content objects, significant performance improvement can be attained over uncoded caching strategies.
\subsection{Main Contribution}~\label{section:12}
In this paper, we study the order-optimal throughput--delay trade-off performance in a large scale content-centric mobile ad hoc network employing {\em subpacketization} in which each node moves according to the {\em reshuffling mobility model}~\cite{alfano} and one central server is able to have access to the whole file library. We assume a cache enabled network in which time is divided into slots and each user requests a content object from the library independently at random according to a Zipf popularity distribution. The most distinctive feature in our model compared to previous approaches is that we consider the case when the users {\em mobility is too fast} to finalize a complete transition of a content in a single time slot. Our model is motivated by the increasing applications involving on-demand high-resolution videos requested by mobile users in future wireless networks. To account for the short time slot duration, we cache the content in multiple segments (subpackets) at the mobile nodes. We present two caching strategies, uncoded and MDS-coded caching. The main technical contributions of this paper are summarized as follows:
\begin{itemize}
	\item We first present a large-scale cache-enabled mobile network framework where the size of each content object is considerably large and thus only a subpacket of a file can be delivered during one time slot.
	\item We characterize fundamental trade-offs between throughput and delay for our content-centric mobile network for both uncoded sequential reception and the MDS-coded random reception cases under the reshuffling mobility model. 
		\item We formulate optimal cache allocation problems (i.e., the optimal content replication strategies) for both uncoded and MDS-coded caching scenarios and characterize the order-optimal solution using Lagrangian optimization.
		\item We analyze the order-optimal throughput--delay trade-off for both uncoded and MDS-coded cases and identify different operating regimes with respect to the transmission range and the number of subpackets. 
	\item We intensively validate our analysis by numerical evaluations including the order-optimal solution to the cache allocation problem and the throughput--delay trade-off.
	\item We identify the case where the performance difference between the uncoded and MDS-coded caching strategies become prominent with respect to system parameters including the Zipf exponent and the number of subpackets in a content object. 
	\item We extend our study to another scenario where each node moves according to the random walk mobility model.
\end{itemize}

The main motivation of the work is to alleviate the problematic case when a network of fast moving entities cannot
be served by the central server or is not cost-effective. For such cases, the idea is to use the cache-aided users as a distributed server for content distribution. This method essentially increases the capacity of the network by using the storage of each mobile node without the deployment of any additional expensive infrastructure. Under our proposed content-centric network, in addition to the caching gain, we are also capable of improving the overall throughput and delay performance since multiple device-to-device (D2D) communications are allowed in a single time slot. This paper is the first attempt to study large-scale content-centric ad hoc networks under a fading mobility model where subpacketization is employed, and thus sheds light on designing a caching framework in such mobility scenarios.
\subsection{Organization}~\label{section:13}
The rest of this paper is organized as follows. In Section II, some prerequisites and the system model is defined. In Section III, the content delivery protocol and reception strategies are presented. In Section IV, the fundamental throughput--delay trade-off is introduced and specialized in terms of scaling laws. The order-optimal throughput--delay trade-offs are derived by introducing the uncoded caching and MDS-coded caching strategies in Sections V and VI, respectively. In Section VII, numerical evaluations are shown to validate our analysis. In Section VIII, our study is extended to the random walk mobility model. Finally, Section IX summarizes the paper with some concluding remarks.
\subsection{Notations}~\label{section:14}
Throughout this paper, $\mathbb{E}[\cdot]$ is the expectation. Unless otherwise stated, all logarithms are assumed to be to the base 2. We use the following asymptotic notation: i) $f(x)= O(g(x))$ means that there exist constants $a$ and $c$ such that $f(x) \leq ag(x)$ for all $x > c$, ii) $f(x) = o(g(x))$ means that  $\lim_{x \rightarrow \infty  }\frac{f(x)}{g(x)} =0 $, iii) $f(x)=\Omega(g(x))$ if $g(x)= O(f(x))$, iv) $f(x) = \omega(g(x))$ means that  $\lim_{x \rightarrow \infty  }\frac{g(x)}{f(x)} =0 $, v)  $f(x)=\Theta(g(x))$ if $f(x)= O(g(x))$  and $f(x)=\Omega(g(x))$~\cite{bigo}.
\section{Prerequisite and System Model}~\label{section:2} 

\subsection{Overview of MDS Coding}~\label{section:21}
Linear coding is among the most popular coding techniques due to its simplicity and performance. The linear coding operation can be summarized as follows. We divide a content file $m$ into $K$ (uncoded) subpackets $\left\{F^{(u)}_{m,1}, F^{(u)}_{m,2}, \cdots, F^{(u)}_{m,K}\right\}$ and transmit them by linearly combining the subpackets with respect to an encoding vector $v = \left\{a_1, \cdots, a_K\right\}$, which is generated over a Galois field $GF(q)$ of size $q$ ~\cite{gf}. Each encoded subpacket is generated by
\begin{equation}\label{eq:enc}
\mathcal{E}_v =\sum_{j=1}^K a_{j} F^{(u)}_{m,j}, 
\end{equation}
\noindent where $\mathcal{E}_v$ is the encoded subpacket corresponding to the encoding vector $v$ and $a_j$ is the encoding coefficient for the $j$th subpacket. In \eqref{eq:enc}, the addition and multiplication operations are performed over the $GF(q)$. In this work, we consider a special class of linear codes, named MDS codes~\cite{mds}. Assume that a content file $m$ is divided into $K$ subpackets that are encoded into $r_m$ coded subpackets \!$\left\{\!F^{(c)}_{m,1},\! \cdots, F^{(c)}_{m,r_m}\!\!\right\}$\! using a $q$-ary $(r_m,K)$ MDS code. Then, by the property of MDS codes, reception of any subset of $K$ MDS-coded subpackets is sufficient to recover the complete file.
\subsection{System Model}~\label{section:22}
Let us consider a content-centric mobile ad hoc network consisting of $n$ mobile nodes and one central server, where $n$ mobile nodes are distributed uniformly at random in the network of a unit area (i.e., the dense network) and the central server is able to have access to the entire library of size $M=\Theta(n^\beta)$ via infinite-speed backhaul, where $0<\beta<1$. The time is divided into independent slots $t_1, t_2,\cdots$, and each mobile node is allowed to initiate a request during its allocated time slot. In our network model, the end nodes are assumed to prefetch a part of (popular) contents in their local cache from the central server when they are indoors. For example, during off-peak times, the central server can initiate the content placement phase and fill the cache of each node. On the other hand, for the case when the actual requests take place, we confine our attention to an outdoor environment where nodes are moving fast. For such cases, since file reception from the central server may not be cost-effective, only D2D communications are utilized for content delivery (e.g. \cite{R1-3}), i.e., the central server does not participate in the delivery phase.

We first adopt the reshuffling model~\cite{alfano} for the nodes' mobility pattern, which assumes that each mobile node will change their position uniformly at random over the network area at the start of each time slot and it will remain static during a time slot. In our content-centric mobile network, each node generates requests for content objects in the library during its allocated time slot. By following the approaches in ~\cite{R1-3, alfano, as_law, jeon, R1-19, liu, d2d}, we assume that the size of each content object $m \in \mathcal{M}$ is the same, where $\mathcal{M}= \left\{1, \cdots, M\right\}$. We assume that every node requests its content object independently according to the Zipf distribution~\cite{alfano,anh,milad,zipf}
\begin{equation}\label{eq:zipf}
p^{pop}_m = \frac{m^{-\alpha}}{H_{\alpha}(M)},
\end{equation}
\noindent where $\alpha>0$ is the Zipf exponent, and $H_{\alpha}(M) = \sum_{i=1}^{M} i^{-\alpha}$ is a normalization constant formed as the Riemann zeta function and is given by
\begin{equation} \label{eq:H} 
H_{\alpha}(M) = \begin{cases}
    \Theta\left(1 \right) & \alpha > 1 \\
    \Theta\left(\log M \right)              & \alpha =1 \\
		\Theta\left( M^{1-\alpha} \right)  & \alpha <1.
\end{cases} 
\end{equation}
The main theme of our study is to understand how to deal with incomplete file transmissions in a mobile network. In such an example, a user watching a high-resolution video on mobile devices may move away from a source node while the file has not been completely transmitted. Simply expanding the time slot to ``fit'' the throughput of the user may not be feasible in such a case since the user is physically moving away resulting in a lost connection. Our goal is to design strategies that are robust against such examples using the concept of {\em subpacketization} and to analyze their performance. Hence, we assume that each content object is divided into $K$=$\Theta(n^\gamma)$ subpackets, where $0\!<\!\gamma\!<\!1$ and every subpacket has the same (unit) size such that each of the requesting nodes is able to completely download one subpacket from one of its nearest source node in one time slot. In a content-centric network, each node is equipped with a local cache to store the content objects, and in our work, we assume a practical scenario where each node is equipped with a local cache of the same finite storage capacity $S\!=\!\Theta(K)$, i.e., the cache can store $S$ distinct subpackets\footnote{Our problem formulation can be extended to a more general case having heterogeneous cache sizes by replacing the total caching constraints in~\eqref{eq:Cons1} and~\eqref{eq:Cons1c} by $\sum_{m=1}^M KX_m \le \sum_{i=1}^n S_i$ and $\sum_{m=1}^M r_m \le \sum_{i=1}^n S_i$, respectively, where $S_i$ is the storage capacity of node $i$. The general problems can be solved by following the same lines as those in Sections \ref{section:5} and \ref{section:6}.}. In cache-enabled wireless networks, content delivery can be divided into two stages, the content placement phase and the content delivery phase. We first describe the placement phase for both uncoded and MDS-coded caching scenarios, which determines the strategy for caching subpackets of content objects in the storage of $n$ nodes.\\
\textbf{Content placement phase for uncoded caching:} Let $X_{m,i}$, $m\in\mathcal{M}$, $i\in\{1,\cdots, K\}:=[1:K]$ represent the number of replicas (will be optimized later on the basis of the popularity of the content $m$) of each subpacket $i$ of content $m$. Since we will assume that $X_{m,i}$ is the same for all $i \in [1:K]$, henceforth we will drop the index $i$ and denote $X_{m,i}$ by $X_{m}$. Similarly as in \cite{as_law, alfano}, during the caching phase, the $X_{m}$ replicas of subpacket $i$ of content object $m$ are stored in the caches of $X_{m}$ distinct nodes. In order to have a feasible cache allocation strategy, $\left\{X_{m}\right\}_{m=1}^M$ should satisfy the following constraints:
\begin{align}
	\sum_{m=1}^MKX_{m} &\leq Sn, \label{eq:Cons1}\\
1 \leq X_{m} &\leq n. \label{eq:Cons2b}
\end{align}
\noindent Note that the total caching constraint in \eqref{eq:Cons1} is a relaxed version of the individual caching constraints~\cite{milad} and the  constraint in \eqref{eq:Cons2b} is to make sure that the network contains at least $1$ and at most $n$ copies of the each content.\\
\textbf{Content placement phase for MDS-coded caching:} 
For the MDS-coded caching strategy, instead of replicating the subpackets, we encode $K$ subpackets of each content $m$  into $r_m$ MDS-coded subpackets (which will be optimized later). During the caching phase, $r_{m}$ encoded subpackets of content object $m$ are stored in the caches of $r_{m}$ distinct nodes. By the property of MDS codes, a client requesting content $m$ only needs to make sure that any $K$ out of $r_m$ distinct MDS-coded subpackets are decoded to successfully recover the entire content $m$. In order to have a feasible cache allocation strategy, $\!\left\{r_{m}\!\right\}_{m=1}^M$\! should satisfy the following constraints: 
\begin{align}
	\sum_{m=1}^Mr_{m} &\leq Sn, \label{eq:Cons1c} \\
	r_{m} &\geq K. \label{eq:Cons2cb}
\end{align}
\noindent Note that the constraint in \eqref{eq:Cons2cb} is to make sure that for some content $m \in \mathcal{M}$, there exist at least $K$ MDS-coded subpackets in the network so that it can be recovered by a requesting node via MDS code decoding.

We now move on to the delivery phase, which allows the requested content objects to be delivered from the source node to the requesting node over wireless channels  (i.e., D2D communications) possibly during peak times. As addressed before, content is assumed to be retrieved under an outdoor environment in which the nodes do not have reliable connection with the central server due to the fast mobility condition. In the delivery phase, each node downloads its requested content object via {\em single-hop} in its allocated time slots\footnote{We note that under the reshuffling mobility model, the network performance cannot be improved by delivering content over multihop routes~\cite[Section III]{alfano}.}, from one of the nodes storing the requested content object in their caches. The protocol model in~\cite{gupta} is adopted for successful content transmission. According to the protocol model, the content delivery from source node $s$ to requesting node $d$ will be successful if the following conditions hold; 1) $d_{sd}(t_a)\!\leq R$ and 2) $d_{bd}(t_a)\geq(1+\Delta)R$, where $d_{sd}(t_a)$ represents the Euclidean distance between the nodes $s$ and $d$ at given time slot $t_a$, $d_{bd}(t_a)$ represents the distance between the nodes $b$ and $d$, for every node $b$ that is simultaneously transmitting at given time slot $t_a$, $\Delta>0$ is a guard factor, and $R>0$ is the transmission range of each node. We assume $R=\Omega\left(\sqrt{\log n/n}\right)$ and $R=O(1)$, such that each square cell of area $a(n)=\!R^2$ has at least one node with high probability (whp) (see~\cite{gupta} for details). When successful transmission occurs, we assume that the total amount of data transferred during the slot is large enough to transfer one subpacket (either uncoded or MDS-coded) of a content from the sender to the receiver\footnote{Unlike our setup, the work in~\cite{algamal} adopted the fluid model to achieve improved performance as the multihop communications become feasible during each slot.}. Nevertheless, in a given time slot, a requesting node can receive no more than one subpacket. Thus, for a requesting node to successfully receive the entire content file, at least $K$ time slots are required. Note that by properly setting the parameter $K$, the size of each subpacket can be flexibly adjusted so that one subpacket would be transmitted or retrieved in one time slot when the above conditions in the protocol model hold.

\subsection{Performance Metrics}~\label{section:23}
In this subsection, we define performance metrics used throughout our paper. We define a {\em scheme} as a sequence of policies, which determines the transmission scheduling in each time slot as well as the cache allocations for all nodes. For a given scheme, the average content transfer delay $D_{avg}(n)$ (expressed in time slots)  and the per-node throughput $\lambda(n)$ (expressed in content/slot) for a content-centric mobile ad hoc network are defined as follows.
\begin{definition}[Average Content Transfer Delay  $D_{avg}(n)$]~\label{df:davg} Let $D(j,i)$ denote the transfer delay of the $i$th request for any content object by node $j$, which is measured from the moment that the requesting message leaves the requesting node until all the $K$ corresponding subpackets of the content object arrives at the node from the source nodes. Then, the delay over all the content requests for node $j$ is $\limsup_{z\rightarrow\infty}\frac{1}{z} \sum_{i=1}^z D(j,i)$ for a particular realization of the network. In this case, the average content transfer delay $D_{avg}(n)$ of all nodes is defined as
\begin{equation} \label{eq:ddavg}
D_{avg}(n) \overset{\Delta}{=} \mathbb{E}\left[ \frac{1}{n}\sum_{j=1}^n \limsup_{z\rightarrow\infty}\frac{1}{z} \sum_{i=1}^z D(j,i)\right],
\end{equation}
where the expectation is over all network realizations.  
\end{definition}
\begin{definition}[Per-Node Throughput $\lambda(n)$]~\label{df:pnt} Let $T(j,\tau)$ denote the total number of requested content objects received by node $j$ during $\tau$ time slots. Note that this could be a random quantity for a given network realization. Then, the per-node throughput $\lambda(n)$ in our cache-enabled mobile network is
\begin{equation} \label{eq:dpnt} 
\lambda(n) \overset{\Delta}{=} \mathbb{E}\left[\frac{1}{n}\sum_{j=1}^n \liminf_{\tau\rightarrow\infty}\frac{1}{\tau} T(j,\tau)\right],
\end{equation}
where the expectation is over all network realizations.
\end{definition}
\section{Content Delivery Protocol and Reception Strategies}~\label{section:3}
In this section, we describe the protocol for the content delivery along with the file reception strategies for both uncoded and MDS-coded caching.
\subsection{Content Delivery}~\label{section:31}
In the following, we explain the strategy for content delivery. First, each node generates a content request for a subpacket (either uncoded or MDS-coded) of content $m$ according to the Zipf's popularity distribution. If the requesting node finds a potential source node within a single-hop range, i.e., within a radius of $R$, then it will start retrieving its desired content. Otherwise, it waits until it finds an available source node for the request. Next, it generates another request for the rest of the subpackets of content $m$ until the requesting node successfully receives the $K$ distinct subpackets. Finally, it generates another request for a new content object by following the same procedure as described above. 
\subsection{Reception Strategies}~\label{section:32}
In this subsection, we explain the sequential and random content reception strategies for uncoded and MDS-coded caching, respectively. These content reception strategies represent the sequence in which the $K$ subpackets of a desired content object are delivered to the requesting node. 
\subsubsection{Sequential Reception of Uncoded Content}~\label{section:321}
We first explain the sequential reception strategy for the uncoded caching case. In the uncoded case, the reception strategy is sequential; that is, all the $K$ subpackets of a content object are delivered in a sequence to the requesting node. An illustration of the sequential reception strategy is shown for three representative cases. In Fig.~\ref{fig:seq_rec_a}, at time slot $t_a$, a node requests the subpacket $F^{(u)}_{m,1}$ of content $m$ to the nodes within its transmission range $R$. The request is respondent within a time slot if there exists a source node that has $F^{(u)}_{m,1}$ in his/her cache and falls within the transmission range of the requesting node at time slot $t_a$. In Fig.~\ref{fig:seq_rec_b}, the requesting node requests the subpacket $F^{(u)}_{m,2}$ of content $m$ at time slot $t_b$ and fails to find any source node within its transmission range. Thus, the requesting node will wait irrespective of the fact that there is a source node within its transmission range storing the subpacket $F^{(u)}_{m,3}$. In Fig.~\ref{fig:seq_rec_c}, the requesting node is still looking for the subpacket $F^{(u)}_{m,2}$ and the request is responded by a source node that has $F^{(u)}_{m,2}$ in his/her cache and falls within the transmission range of the requesting node at time slot $t_c$.
\begin{figure}[t]
\centering
\subfigure[Time Slot $t_a$]{  \includegraphics[ height=2.7cm, width= 2.6cm]{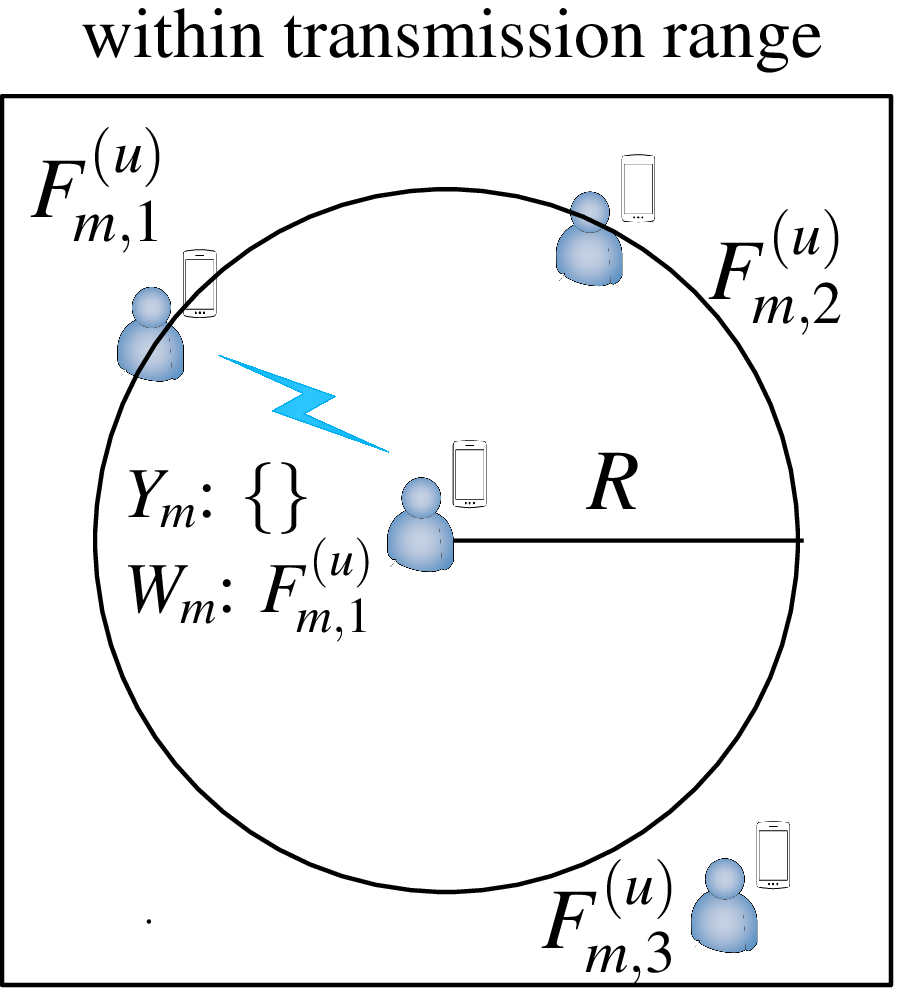} \label{fig:seq_rec_a}}
\subfigure[Time Slot $t_b$]{  \includegraphics[ height=2.7cm, width= 2.6cm]{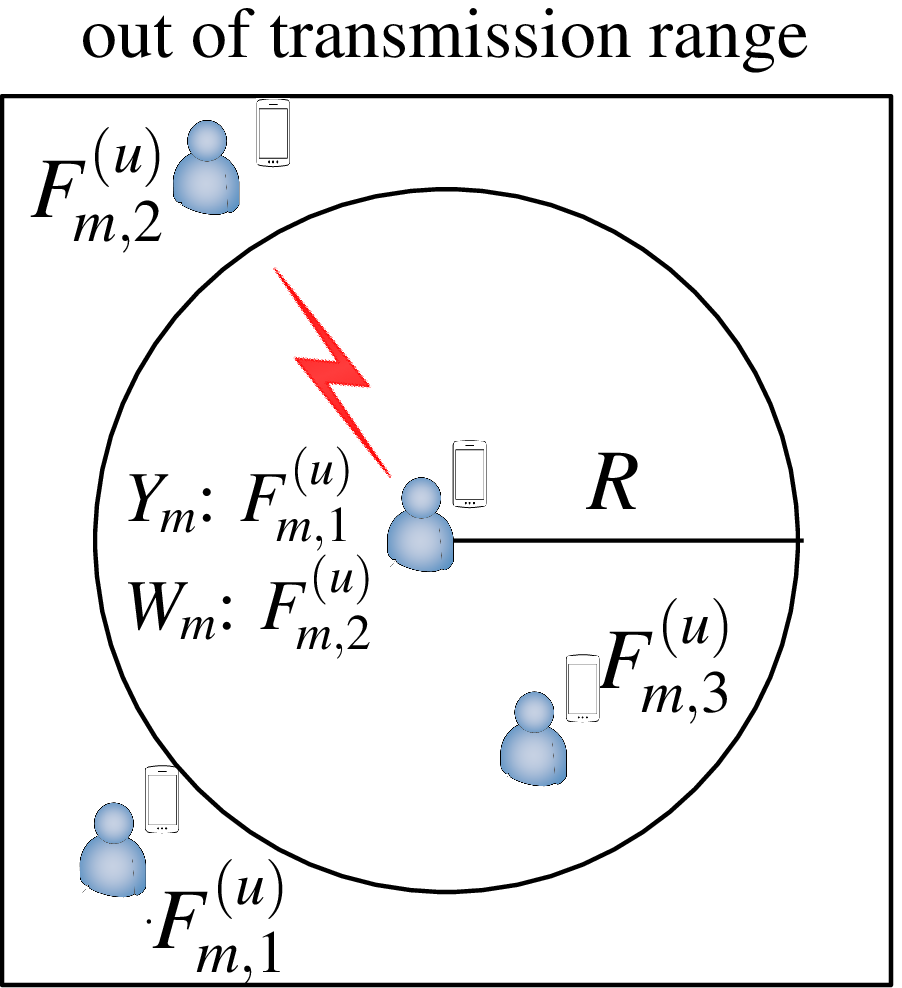} \label{fig:seq_rec_b}}
\subfigure[Time Slot $t_c$]{  \includegraphics[ height=2.7cm, width= 2.6cm]{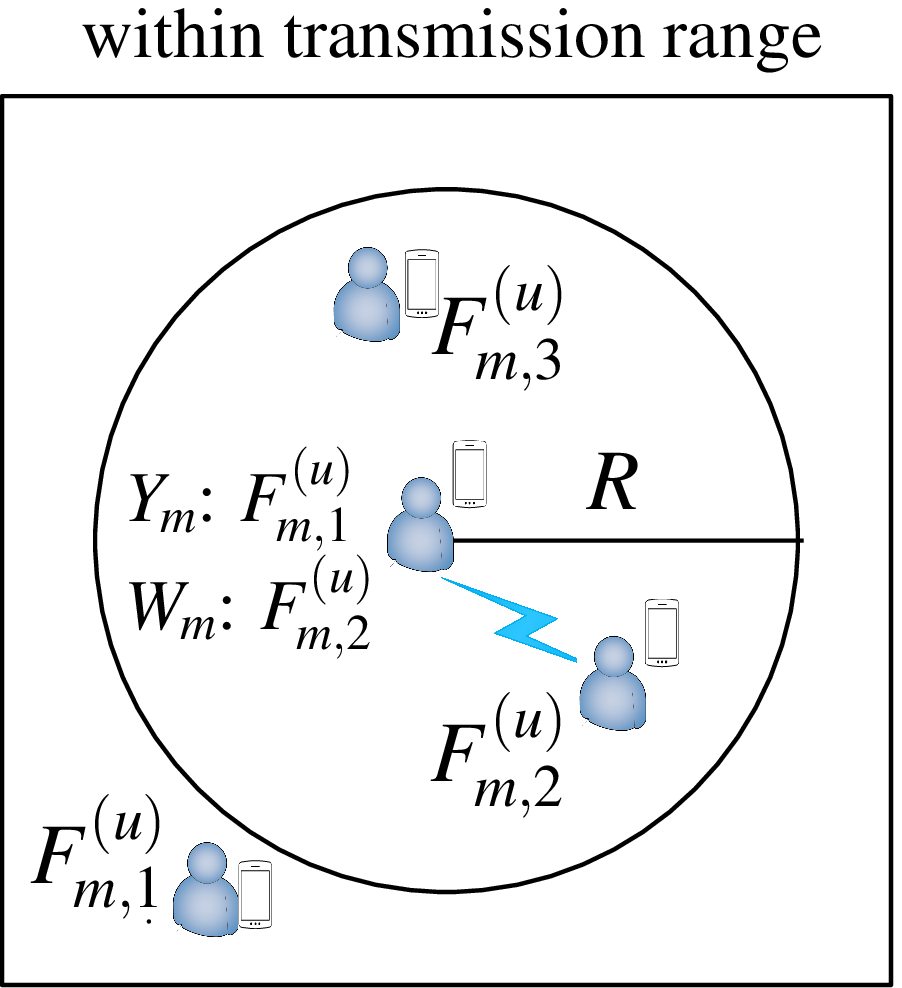} \label{fig:seq_rec_c}}
\caption{ Content delivery following the sequential reception strategy for an uncoded content, where $F^{(u)}_{m,i}$ is the $i$th uncoded subpacket of content $m$, $Y_m$ is the set of received content's subpackets, and $W_m$ is the set of required subpackets of content $m$.}
\label{fig:seq_rec}
\end{figure} 
\subsubsection{Random Reception of MDS-coded Content}~\label{section:322}
In the MDS-coded caching case, file reception is random; that is, the requesting node may receive any of the $K$ out of $r_m$ MDS-coded subpackets of a content object in an arbitrary order. Figure~\ref{fig:ran_rec} is an illustration of the random reception strategy. In Fig.~\ref{fig:ran_rec_a}, at time slot $t_a$, a node requests the subpackets of content $m$ from the nodes within its transmission range $R$. The request is respondent within the one time slot by a source node that has $F^{(c)}_{m,2}$ in his/her cache and falls within the transmission range of the requesting node at time slot $t_a$. In Fig.~\ref{fig:ran_rec_b}, a node requests the remaining subpackets of content $m$ and the request is responded by a source node that has $F^{(c)}_{m,3}$ in his/her cache that falls within the transmission range of the requesting node at time slot $t_b$.
\begin{figure}[t]
 \centering
 \subfigure[Time Slot $t_a$ ] { \includegraphics[height=3.8cm]{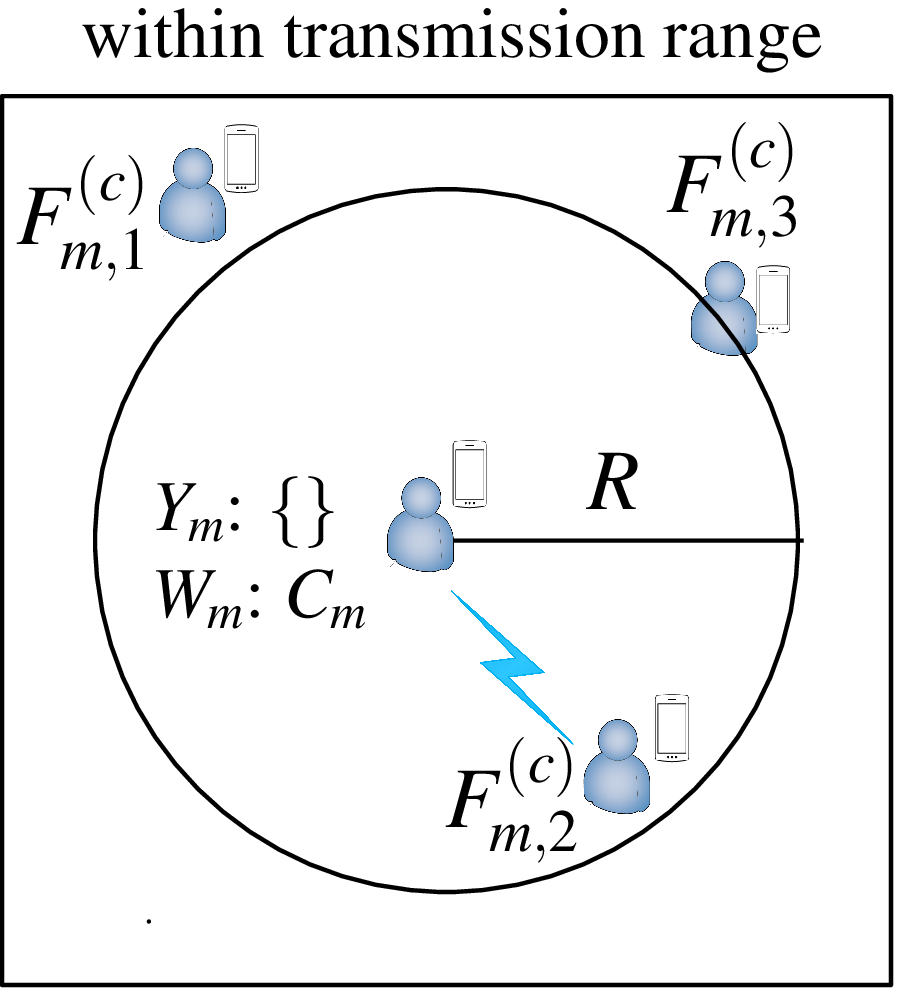} \label{fig:ran_rec_a} }
 \subfigure[Time Slot $t_b$ ] { \includegraphics[height=3.8cm]{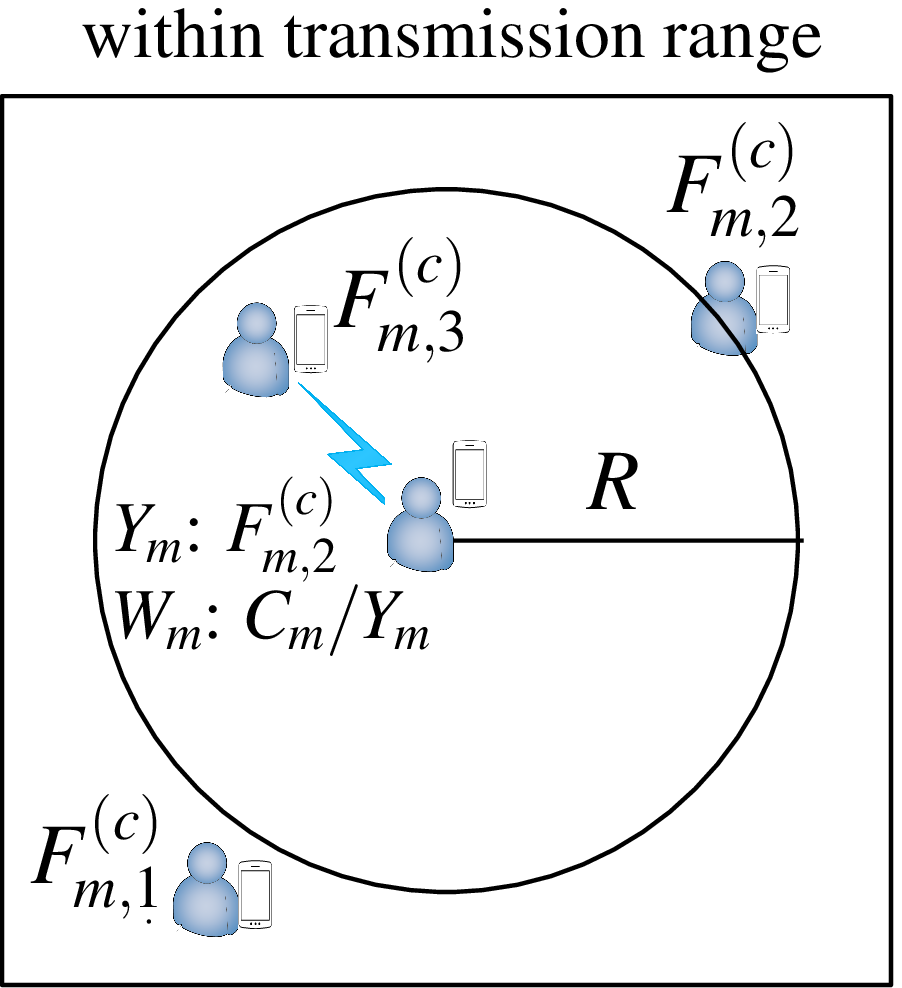} \label{fig:ran_rec_b}	}
\caption{ Content delivery following the random reception strategy for an MDS-coded content, where $F^{(c)}_{m,j}$ is the $j$th MDS-coded subpacket of content $m$, $C_m$ is the set of all the MDS-coded subpackets of content $m$, $Y_m$ is the set of received content's MDS-coded subpackets, and $W_m$ is the set of required MDS-coded subpackets of content $m$.}
\label{fig:ran_rec}
\end{figure}
Intuitively, the random reception strategy should perform better than the sequential reception case. Nonetheless, both schemes play an important role in caching for different applications in practice. For example, the random reception strategy seems to be suitable for the case where content such as videos and documents are first downloaded completely, and then viewed offline. On the other hand, for the case where a user is streaming videos online, the random reception strategy will not work since the content is required to be downloaded sequentially.

 We note that a playback buffer~\cite{buffer} that stores a few future subpackets could enhance the quality of a video streaming service as it enables us to play the next portion of the video, e.g., $j$th subpacket $F_{m,j}^{(u)}$  of content $m\in\mathcal{M}$ before a requesting user reaches the point of viewing the end of the current subpacket of the video (i.e., $F_{m,j-1}^{(u)}$). In this paper, we will not account for how such playback buffers are deployed and how the content is updated within the playback buffer, which goes beyond our scope\footnote{As long as the playback buffer has limited capacity that is independent of the scaling system parameters, content delivery for online on-demand video streaming would only be possible by sequential reception, which is consistent with the current video streaming protocols as discussed in~\cite{R1-19, R1-3}. This is because buffering the subpackets in an arbitrarily way does not guarantee seamless video streaming as the buffers may not contain the next sequential portion of the video that a requesting node is watching.}.

We will see in the next section that these reception strategies play key roles in defining the average content transfer delay $D_{avg}(n)$.
\section{Throughput--Delay Trade-off}~\label{section:4}
In this section, we characterize a fundamental throughput--delay trade-off in terms of scaling laws for the content-centric mobile network using the proposed content delivery protocol. 
\begin{theorem}\label{th:TD} 
Consider nodes generating requests according to the content delivery protocol in Section \ref{section:31}. Then, the throughput--delay trade-off in our proposed cache-enabled mobile network is given by
\begin{align} \label{eq:tdo} 
\lambda(n) = \Theta\left(\frac{1}{na(n) D_{avg}(n)}\right),
\end{align}
where $\lambda(n)$ is the per-node throughput, $D_{avg}(n)$ is the average content transfer delay, and $a(n)=R^2$ is the area in which a node can communicate with other nodes. 
\end{theorem}
\begin{proof}  
The fundamental throughput--delay trade-off for the content-centric network employing the proposed content delivery protocol in Section~\ref{section:31} can be established using the elementary renewal theorem~\cite[Chapter 8]{renewal}. Let $\kappa(\tau, j)$ denote the total number of content objects transferred to request node $j$ observed up to $\tau$ time slots when node $j$ is assumed to be an active requester in every time slot. Then from the fact that the transfer delay $D(j,i)$ of the $i$th request for any content object by node $j$ represents the inter-arrival time, it follows whp that
\begin{align*} 
\frac{1}{n}\sum_{j=1}^n \lim_{\tau\rightarrow\infty} \frac{\kappa(\tau,j)}{\tau} = \frac{1}{D_{avg}(n)},
\end{align*}
where $D_{avg}(n)$ is the average content transfer delay over all nodes in \eqref{eq:ddavg}. Since only one node in the transmission range of area $a(n)$ can be active in each time slot, the achievable per-node throughput in \eqref{eq:dpnt} is then expressed as \eqref{eq:tdo}, which completes the proof of Theorem~\ref{th:TD}.
\end{proof}
Theorem~\ref{th:TD} implies that the per-node throughput $\lambda(n)$ can be characterized for given the average content transfer delay $D_{avg}(n)$ or vice versa. Hence, we focus on minimizing $D_{avg}(n)$, which is equivalent to maximizing $\lambda(n)$ for a given $a(n)$. We establish the following lemma, which formulates the average content transfer delay $D_{avg}(n)$ for both the uncoded sequential reception in Section~\ref{section:321} and the MDS-coded random reception in Section~\ref{section:322}.
\begin{lemma}\label{le:otd} 
Consider a content-centric mobile network with nodes retrieving their requests according to the content delivery protocol in Section~\ref{section:31}. Given a cache allocation strategy, the average content transfer delay $D_{avg}(n)$ for the uncoded caching case employing the sequential reception strategy in Section~\ref{section:321} is given by
\begin{align} \label{eq:davg_s} 
D_{avg}(n) = \Theta \left( \sum_{m=1}^M   \frac{K p^{pop}_m}{\min\left( 1, a(n) X_{m} \right)} \right)
\end{align}
\noindent and $D_{avg}(n)$ for the MDS-coded caching case employing the random reception strategy in Section~\ref{section:322} is given by
\begin{align} \label{eq:davg_c} 
D_{avg}(n) = \Theta \left( \sum_{m=1}^M  \sum_{j=0}^{K-1} \frac{p^{pop}_m}{\min\left( 1, (r_m-j) a(n) \right)} \right).
\end{align}
\end{lemma}
\begin{proof} First, consider the uncoded caching case employing the sequential reception strategy. Given a cache allocation strategy $\left\{X_{m}\right\}_{m=1}^M$, for any requesting mobile node, the transfer delay associated to the $i$th subpacket of content $m\in \mathcal{M} $ is given by the number of time slots that it take for a node to come in contact with another node storing the desired content, which is geometrically distributed with mean $1/p^{seq}_{m,i}$. Here, $p^{seq}_{m,i}$ is the contact probability that a node requesting the $i$th subpacket of content $m \in \mathcal{M}$ falls in a given time slot within distance $R$ of a node holding the requested subpacket, which is given by
\begin{align} \label{eq:pmi} 
 p^{seq}_{m,i}= 1- \left( 1- a(n)\right)^{X_{m}} ~~~ i\in\left[ 1:K\right].
\end{align}
The contact probability $p^{seq}_{m,i}$ in order sense is equivalent to $\Theta\left( \min\left( 1, a(n) X_{m} \right)\right)$. Then, the number of time slots required to successfully receives content object $m\in \mathcal{M}$ consisting of $K$ subpackets is given by $\Theta \left(\frac{K}{\min\left( 1, a(n) X_{m} \right)} \right)$. From the fact that each node generates its request following the same Zipf's law, the $D_{avg}(n)$ for the content-centric mobile network employing the sequential reception of the uncoded content is given by   
\begin{equation*} \label{eq:davg_S} 
D_{avg}(n) = \Theta \left( \sum_{m=1}^M \frac{Kp^{pop}_m}{\min\left( 1, a(n) X_{m} \right)}\right).    
\end{equation*} 
Next, we characterize the average content transfer delay $D_{avg}(n)$ for the case of MDS-coded caching employing the random reception strategy. Given a cache allocation strategy $\left\{r_{m}\right\}_{m=1}^M$, the contact probability $p^{ran}_{m,j}$ for the MDS-coded caching based random reception strategy is the probability that a node having pending requests for $K-j$ MDS-coded subpackets of content $m$ falls in a given time slot within distance $R$ of a node holding one of the requested MDS-coded subpackets while the requesting node is assumed to have already received $j$ MDS-coded subpackets. Then, $p^{ran}_{m,j}$ is given by 
\begin{equation} \label{eq:pmjc} 
p^{ran}_{m,j} = 1- \left( 1- a(n)\right)^{(r_m-j)} ~~~j \in\left[ 0:K-1\right].
\end{equation}
The contact probability $p^{ran}_{m,j}$ in order sense is equivalent to $\Theta\left(\min\left( 1, (r_m-j) a(n) \right)\right)$. Then, the expected number of time slots required to successfully receives content object $m\in \mathcal{M}$ consisting of $K$ MDS-coded subpackets is given by $\Theta \left( \sum_{j=0}^{K-1} \frac{1}{\min\left( 1, (r_m-j) a(n) \right)} \right)$. Thus, the $D_{avg}(n)$ for the content-centric mobile network employing the random reception of the MDS-coded content is given by
\begin{equation*} \label{eq:davg_C} 
D_{avg}(n) = \Theta \left( \sum_{m=1}^M  \sum_{j=0}^{K-1} \frac{p^{pop}_m}{\min\left( 1, (r_m-j) a(n) \right)} \right).
\end{equation*} 
This completes the proof of the lemma.
\end{proof}
From Theorem~\ref{th:TD} and Lemma~\ref{le:otd}, the per-node throughput $\lambda(n)$ for the case of uncoded caching can be obtained using~\eqref{eq:tdo} and~\eqref{eq:davg_s}, while the per-node throughput $\lambda(n)$ for the case of MDS-coded caching can be obtained using~\eqref{eq:tdo} and~\eqref{eq:davg_c}. As expected, Lemma~\ref{le:otd} implies that the average content transfer delay $D_{avg}(n)$ for both reception strategies is influenced by the cache allocation strategies. The optimal performance in term of minimum average content transfer delay $D_{avg}(n)$ can be obtained by optimally selecting the cache allocation strategy, which is not straightforward due to caching constraints. Also, note that by Theorem~\ref{th:TD}, selecting the optimal cache allocation strategy that minimizes $D_{avg}(n)$ is equivalent to maximizing $\lambda(n)$ for a given $a(n)$. In the next section, we characterize the minimum average content transfer delay $D_{avg}(n)$ under our network model with subpacketization for uncoded caching by presenting the optimal cache allocation strategy.
%%%%%%%%%%%%%%%%%%%%%%%%%%%%%%%%%%%%%%%%%%%%%%%%%%%%%%%%%%%%%%%%%%%%%%%%%%%%%%%%%%%%%%%%%%%%%%%%%%%%%%%%%%%%%%%%%%%%%%%%%%%%%%%%%%%%%%%%%%%%%%%%%%%%%%%%%%%%%%%%%%%%%%%%%%%%%%%%%%%%%%%%%%%%%%%%%%%%%%%%%%%%%%%%%%%%%%%%%%%%%%%%%%%%%%%%%%%%%%%%%%%%%%%%%%%%%%%%%%%%%%%%%%%%%%%%%%%%%%%%%%%%%%%%%%%%%%%%%%%%%%%%%%%%%%%%%%%%%%%%%NEW%SECTION%5%%%%%%%%%%%%%%%%%%%%%%%%%%%%%%%%%%%%%%%%%%%%%%%%%%%%%%%%%%%%%%%%%%%%%%%%%%%%%%%%%%%%%%%%%%%%%%%%%%%%%%%%%%%%%%%%%%%%%%%%%%%%%%%%%%%%%%%%%%%%%%%%%%%%%%%%%%%%%%%%%%%%%%%%%%%%%%%%%%%%%%%%%%%%%%%%%%%%%%%%%%%%%%%%%%%%%%%%%%%%%%%%%%%%%%%%%%%%%%%%%%%%%%%%%%%%%%%%%%%%%%%%%%%%%%%%%%%%%%%%%%%%%%%%%%%%%%%%%%%%%%%%%%%%%%%%%%%%%%%%%%%%%%%%%%%%%%%%%%%%%%%%%%%%%%%%%%%%%%%%%
\section{Order-Optimal Uncoded Caching in Mobile Networks with Subpacketization }~\label{section:5}
In this section, we characterize the order-optimal average content transfer delay $D_{avg}(n)$ and the corresponding maximum per-node throughput $\lambda(n)$ of the cache-enabled mobile ad hoc network employing subpacketization by selecting the order-optimal cache allocation strategies $\{\hat{X}_m\}_{m=1}^M$. We first introduce our problem formulation in terms of minimizing the average content transfer delay $D_{avg}(n)$ for the uncoded caching following the sequential reception strategy in Section~\ref{section:321}. Then, we solve the optimization problem and present the order-optimal cache allocation strategy under our network model. Finally, we present the minimum $D_{avg}(n)$ and the corresponding maximum $\lambda(n)$ using the order-optimal cache allocation strategy.
\subsection{Problem Formulation}~\label{section:51}
It is observed from Lemma~\ref{le:otd} that the average content transfer delay $D_{avg}(n)$ depends completely on the caching allocation strategy $\{X_m\}_{m=1}^M$. Among all the cache allocation strategies, the optimal one will be the one that has the minimum $D_{avg}(n)$. It is intuitive that there is no need to cache more than $a(n)^{-1}$ replicas of the subpacket $i$ of content object $m \in \mathcal{M}$ over the network for the uncoded sequential reception case due to the term $\min\left( 1, a(n) X_{m} \right)$ in~\eqref{eq:davg_s} of Lemma~\ref{le:otd}. Thus, we modify \eqref{eq:Cons2b} and impose the following individual caching constraints:
\begin{align} \label{eq:Cons2} 
1 \leq X_{m} \leq a(n)^{-1}
\end{align}
for all $m \in \mathcal{M}$.
Now, from Lemma~\ref{le:otd} and the caching constraints in \eqref{eq:Cons1} and \eqref{eq:Cons2}, the optimal cache allocation strategy $\{\hat{X}_{m}\}_{m = 1}^M$ for the uncoded sequential reception scenario can thus be the solution to the following optimization problem: 
\begin{subequations} \label{eq:op_s}
\begin{align} \label{eq:of_s}
\min_{\left\{X_{m}\right\}_{m \in \mathcal{M}}}\sum_{m=1}^M \frac{Kp^{pop}_{m}}{a(n) X_{m}}
\end{align}
subject to 
 \begin{align} \label{eq:c1}
  \sum_{m=1}^MKX_{m} \leq Sn,
\end{align}
\begin{align} \label{eq:c2}
1 \leq X_{m} \leq a(n)^{-1}.
\end{align}
\end{subequations}
Note that the number of replicas $X_{m}$ of content object $m$ stored at the mobile nodes is an integer variable, which makes the optimization problem~\eqref{eq:op_s} non-convex and thus intractable. However, as long as scaling laws are concerned, the discrete variables $X_{m}$ for $m \in \mathcal{M}$ can be relaxed to real numbers in $[1, \infty)$ so that the objective function in \eqref{eq:op_s} becomes convex and differentiable.
\subsection{Order-Optimal Cache Allocation Strategy}~\label{section:52}
We use the Lagrangian method to solve the problem in~\eqref{eq:op_s}. Before diving into the optimization problem, we will introduce some useful operating regimes. In particular, we divide the entire content domain $\mathcal{M}$ into the following regimes according to content $m\in \mathcal{M}$:
\begin{itemize}
    \item Regime $\text{I}^{(u)}$: $X_{m}=\Theta\left( a(n)^{-1} \right)$
    \item Regime $\text{II}^{(u)}$: $ X_{m} = o\left( a(n)^{-1} \right)$.
\end{itemize}
\noindent Let $\mathcal{I}^{(u)}_{1}$ and $\mathcal{I}^{(u)}_{2}$ be partitions of $\mathcal{M}$ that consist of content belonging to Regimes I\(^{(u)}\)  and II\(^{(u)}\), respectively. The Lagrangian function corresponding to \eqref{eq:op_s} by relaxing the $1 \leq X_{m}$ constraint is given by 
\begin{align}\label{eq:LF_s}
&\mathcal{L}\left(\!\left\{X_{m}\right\}_{m \in \mathcal{M}},\delta,\left\{\sigma_m\right\}_{m \in \mathcal{M}}\right)\! =\!  \sum_{m=1}^M\!\! \frac{K p^{pop}_{m}}{a(n)X_{m}}\!\nonumber \\ & +\! \delta\left(\sum_{m= 1}^{M}\!\!KX_{m}\!- Sn\!\right) + \sum_{m= 1}^{M}\!\!\sigma_m \left( X_{m}\!-\frac{1}{a(n)}\right),
\end{align}
\noindent where $\sigma_m, \delta \in \mathbb{R}$. The Karush-Kuhn-Tucker (KKT) conditions for \eqref{eq:op_s} are then given by
\begin{align}\label{eq:KKT1_s}
\frac{\partial \mathcal{L}\left(\left\{\hat{X}_{m}\right\}_{m \in \mathcal{M}},\hat{\delta},\left\{\hat{\sigma}_m\right\}_{m \in \mathcal{M}}\right) }{\partial \hat{X}_{m}}= 0,  
\end{align}
\begin{align}\label{eq:KKT2_s}
\hat{\sigma}_m \left( \hat{X}_{m}- a(n)^{-1} \right) =0 , 
\end{align}
\begin{align}\label{eq:KKT3_s}
\hat{\delta} \left(\sum_{m= 1}^{M} K\hat{X}_{m} - Sn \right) =0 ,
\end{align}
\begin{align*}%\label{eq:KKT4_s}
\hat{\delta} \geq 0 ,
\end{align*}
\begin{align*}%\label{eq:KKT5_s}
\hat{\sigma}_m \geq 0
\end{align*}
\noindent for all $m \in \mathcal{M}$, where $\hat{X}_m, \hat{\delta}$, and $\hat{\sigma}_m$ represent the optimized values. Let the content index $m^{(u)}_1\in \mathcal{I}^{(u)}_{2}$ denote the smallest content index belonging to Regime II\(^{(u)}\). In the following, we introduce a lemma that presents an important characteristic of the optimal cache allocation strategy $\left\{\hat{X}_{m}\right\}_{m=1}^M$ and  plays a vital role in solving \eqref{eq:op_s}.
\begin{lemma}\label{le:1_s}
The order-optimal cache allocation strategy denoted by $\left\{\hat{X}_{m}\right\}_{m=1}^M$ in \eqref{eq:op_s} is non-increasing with $m \in \mathcal{M}$.
\end{lemma}
\begin{proof} Deferred to Appendix \ref{AppendixA_s}. 
\end{proof}
Lemma \ref{le:1_s} allows us to establish our first main result regarding the order-optimal cache allocation strategy for the uncoded case.
\begin{prop}\label{th:oprep_s}
Consider the content-centric mobile ad hoc network model employing subpacketization and following the uncoded sequential reception strategy in Section~{\em\ref{section:321}}. The order-optimal cache allocation strategy is then given by
\begin{equation} \label{eq:oprep_s} 
\hat{X}_m = \begin{cases}
    a(n)^{-1} &  m \in \left\{1, \cdots, m^{(u)}_1-1 \right\} \\
		 \frac{\sqrt{p^{pop}_m}}{\sum_{\widetilde{m}= m^{(u)}_1}^M \sqrt{p^{pop}_{\widetilde{m}}} } S^{(u)} &  m \in \left\{m^{(u)}_1, \cdots, M \right\}
\end{cases} 
\end{equation}
\noindent where $p^{pop}_m$ is given in~\eqref{eq:zipf},  $S^{(u)}= n - (m^{(u)}_1-1) a(n)^{-1}$, and the boundary between Regimes {\em I}\(^{(u)}\) and {\em II}\(^{(u)}\) is defined by content index $m^{(u)}_1$, which is given by
\begin{equation} \label{eq:am1_s}
m^{(u)}_1 =\Theta\left( \min\left\{M, \left(\frac{na(n)}{H_{\frac{\alpha}{2}}(M)}\right)^{2/\alpha}  \right\}\right),
\end{equation}
where
\begin{equation} \label{eq:H2} 
H_{\frac{\alpha}{2}}\left(M\right) = \begin{cases}
    \Theta\left(1 \right) & \alpha > 2 \\
    \Theta\left(\log M \right)              & \alpha =2 \\
		\Theta\left( M^{1-\frac{\alpha}{2}} \right)  & \alpha <2.
\end{cases} 
\end{equation}
\end{prop}
\begin{proof}
Deferred to Appendix \ref{AppendixD_s}.
\end{proof}
From Proposition~\ref{th:oprep_s}, it is observed that the order-optimal cache allocation strategy is partitioned into two parts. The first part consisting of highly popular content with indice $m <m^{(u)}_1$ is replicated $a(n)^{-1}$ times. The rest is the content with index $m \geq m^{(u)}_1$ for which the order-optimal cache allocation strategy is to monotonically decrease the number of replicas with $m$. In addition, the value of $m^{(u)}_1$ depends on the choice of $a(n)$ and the Zipf exponent $\alpha$.

 Next, based on our uncoded cache allocation strategy for the total caching constraint in~\eqref{eq:op_s}, we extend the strategy to satisfy the local caching constraints. Based on the solution $\{\hat{X}_{m}\}_{m=1}^M $ in Proposition~\ref{th:oprep_s}, the central server places replicas of the content in the cache of each node according to the replica allocation algorithm in~\cite[Appendix C]{alfano}, in which contents are considered in sequence and the algorithm is decomposed into $MK$ steps. The design of this algorithm is basically inspired by the well-known water-filling strategy. Each $(m,k)$th step (i.e., the $\left(k+(m-1)K\right)$th step) of the algorithm is responsible for caching the $\left\lceil \hat{X}_m\right\rceil$ replicas of the $k$th subpacket of content $m\in \mathcal{M}$. Here, $\left\lceil x\right\rceil$ denotes the ceiling function of $x$. More specifically, a set of $\left\lceil \hat{X}_m\right\rceil$ distinct nodes $\mathcal{N}_{m,k}^{(u)}$ is selected and a replica of the $k$th subpacket of content $m$ is assigned to each node in the set $\mathcal{N}^{(u)}_{m,k}$ at the $(m,k)$th step of the algorithm. In the first step (i.e., the $(1,1)$th step), $\left\lceil \hat{X}_1 \right\rceil$ nodes are randomly assigned to the set $\mathcal{N}_{1,1}^{(u)}$. In the subsequent process, at $(m,k)$th step, first all nodes are sorted in ascending order of the total number of subpackets cached by each node since the algorithm has been initiated, and then the top-$\left\lceil \hat{X}_m \right\rceil$ nodes from the sorted list are assigned to the set $\mathcal{N}_{m,k}^{(u)}$. In other words, a preference is given to the nodes to cache replicas in terms of the number of assigned subpackets to date. If there is a tie in the number of subpackets assigned to users' caches after sorting of each step, then a random node selection is made. The above steps are repeated $MK$ times until all the replicas of the content are assigned.
\begin{remark}\label{R:1} 
Due to the fact that $\sum_{m=1}^M K\left\lceil \hat{X}_m \right\rceil \leq2 \sum_{m=1}^M$ $ K \hat{X}_m \leq 2Sn$, it is shown that as far as the cache of each node is filled with replicas according to the above replica allocation algorithm, the proposed order-optimal cache allocation strategy in Proposition~\ref{th:oprep_s} can be extended to satisfy the property that the number of subpackets stored by each node (i.e., the storage capacity per node) is bounded by $2S$, which is given by $\Theta(K)$. Hence, our cache allocation strategy in Proposition~\ref{th:oprep_s} fulfills the local cache size constraints within a factor of 2. 
\end{remark}
In the next subsection, we characterize the optimized minimum average content transfer delay $D_{avg}(n)$ by adopting the order-optimal cache allocation strategy presented in Proposition~\ref{th:oprep_s} and analyze the impact of some key parameters, $K$, $M$, $a(n)$, and $\alpha$ on the order-optimal performance. 
\subsection{Order-Optimal Performance}~\label{section:54}
In this subsection, we compute the minimum average content transfer delay $D_{avg}(n)$ using the order-optimal cache allocation strategy obtained in Proposition~\ref{th:oprep_s}.
\begin{theorem}\label{th:delay_s} Consider a content-centric mobile ad hoc network model with subpacketization adopting the order-optimal cache allocation strategy  $\{\hat{X}_{m}\}_{m=1}^M$ in~\eqref{eq:oprep_s} and following the uncoded sequential reception strategy. Then, the minimum average content transfer delay $D_{avg}(n)$ is given by
\smallskip
\begin{equation} \label{eq:delay_s2} 
D_{avg}(n) = \Theta\left(\max\left\{K, \frac{K\left(H_{\frac{\alpha}{2}}(M)\right)^2}{na(n)H_{\alpha}(M)} \right\} \right), 
\end{equation}
\noindent where $K$ is the number subpackets of each content, $a(n)$ is the area in which a node can communicate with other nodes, and $H_\alpha(M)$ and $H_\frac{\alpha}{2}(M)$ are given in \eqref{eq:H} and \eqref{eq:H2}, respectively. 
\end{theorem}
\smallskip
\begin{proof}
Deferred to Appendix \ref{AppendixE_s}.
\end{proof}
From Theorems~\ref{th:TD} and~\ref{th:delay_s}, the maximum achievable per-node throughput $\lambda(n)$ is given by
\begin{equation}\label{eq:pnt_s}
\lambda(n)= \Theta\left(\min\left\{\frac{1}{na(n)K} , \frac{H_{\alpha}(M)}{K\left(H_{\frac{\alpha}{2}}(M)\right)^2} \right\}\right).
\end{equation}
If the content follows the Zipf distribution with exponent $\alpha > 2$, then the best delay $D_{avg}(n)=\Theta(K)$ and the corresponding throughput $\lambda(n)=\Theta\left(\frac{1}{na(n)K}\right)$ are achieved. When $\alpha \leq 2$, the minimum delay $D_{avg}(n)$ and the corresponding throughput $\lambda(n)$ start to scale with $a(n)$, $K$, and $M$. In the next section, we characterize the minimum average content transfer delay $D_{avg}(n)$ and the corresponding per-node throughput $\lambda(n)$ under our network model for the MDS-coded caching case by presenting the order-optimal cache allocation strategy.
%%%%%%%%%%%%%%%%%%%%%%%%%%%%%%%%%%%%%%%%%%%%%%%%%%%%%%%%%%%%%%%%%%%%%%%%%%%%%%%%%%%%%%%%%%%%%%%%%%%%%%%%%%%%%%%%%%%%%%%%%%%%%%%%%%%%%%%%%%%%%%%%%%%%%%%%%%%%%%%%%%%%%%%%%%%%%%%%%%%%%%%%%%%%%%%%%%%%%%%%%%%%%%%%%%%%%%%%%%%%%%%%%%%%%%%%%%%%%%%%%%%%%%%%%%%%%%%%%%%%%%%%%%%%%%%%%%%%%%%%%%%%%%%%%%%%%%%%%%%%%%%%%%%%%%%%%%%%%%%%%%NEW%SECTION%6%%%%%%%%%%%%%%%%%%%%%%%%%%%%%%%%%%%%%%%%%%%%%%%%%%%%%%%%%%%%%%%%%%%%%%%%%%%%%%%%%%%%%%%%%%%%%%%%%%%%%%%%%%%%%%%%%%%%%%%%%%%%%%%%%%%%%%%%%%%%%%%%%%%%%%%%%%%%%%%%%%%%%%%%%%%%%%%%%%%%%%%%%%%%%%%%%%%%%%%%%%%%%%%%%%%%%%%%%%%%%%%%%%%%%%%%%%%%%%%%%%%%%%%%%%%%%%%%%%%%%%%%%%%%%%%%%%%%%%%%%%%%%%%%%%%%%%%%%%%%%%%%%%%%%%%%%%%%%%%%%%%%%%%%%%%%%%%%%%%%%%%%%%%%%%%%%%%%%%%%
\section{Order-Optimal MDS-coded Caching in Mobile Networks with Subpacketization }~\label{section:6}
In this section, we propose the order-optimal MDS-coded cache allocation strategies $\{\hat{r}_m\}_{m=1}^M$ to characterize the order-optimal average content transfer delay $D_{avg}(n)$ and the corresponding maximum per-node throughput $\lambda(n)$ of the cache-enabled mobile ad hoc network employing subpacketization. We first introduce our problem formulation in terms of minimizing $D_{avg}(n)$ for the MDS-coded caching following the random reception strategy in Section~\ref{section:322}. Then, we solve the optimization problem and propose the order-optimal cache allocation strategy $\left\{\hat{r}_m\right\}_{m=1}^M$ under our network model. Finally, we present the minimum $D_{avg}(n)$ and the corresponding maximum $\lambda(n)$ using the order-optimal cache allocation strategy.
\subsection{Problem Formulation}~\label{section:61}
It can be seen that there is no need to cache more than $a(n)^{-1}+K$ MDS-coded subpackets of content object $m \in \mathcal{M}$ over the network for the MDS-coded random reception case due to the term $\min\left( 1, (r_m-j) a(n) \right)$ in~\eqref{eq:davg_c} of Lemma~\ref{le:otd}. Thus, we modify \eqref{eq:Cons2cb} and impose the following individual caching constraints:
\begin{equation} \label{eq:Cons2c} 
K \leq r_{m} \leq a(n)^{-1}+K
\end{equation}
for all $m \in \mathcal{M}$. Now, from Lemma~\ref{le:otd} and the caching constraints in \eqref{eq:Cons1c} and \eqref{eq:Cons2c}, the optimal cache allocation strategy $\left\{\hat{r}_{m}\right\}_{m = 1}^M$ for the MDS-coded random reception scenario can thus be the solution to the following optimization problem: 
\begin{subequations} \label{eq:opc}
\begin{equation} \label{eq:ofc}
\min_{\left\{r_{m}\right\}_{m \in \mathcal{M}}}\sum_{m=1}^M  \sum_{j=0}^{K-1} \frac{p^{pop}_{m}}{\min\left( 1, (r_m-j) a(n) \right)}
\end{equation}
subject to
\begin{equation} \label{eq:c1c}
 \sum_{m=1}^Mr_{m} \leq Sn,
\end{equation}
\begin{equation} \label{eq:c2c}
K \leq r_{m} \leq a(n)^{-1}+K.
\end{equation}
\end{subequations}
Similarly as in uncoded caching case, we relax the discrete variables $r_{m}$ for $m \in \mathcal{M}$ to real numbers in $[K, \infty)$ so that the objective function in \eqref{eq:opc} becomes convex and differentiable.
\subsection{Order-Optimal Cache Allocation Strategy}~\label{section:62}
The objective function in \eqref{eq:ofc} contains a $\min$ function in the denominator, which makes the optimization problem intractable. Thus, we first simplify the objective function in \eqref{eq:ofc} and then solve the simplified optimization problem to obtain the order-optimal cache allocation strategy.
\subsubsection{Simplifying Objective Function}~\label{section:621} We simplify the objective function in \eqref{eq:ofc} by dividing the entire content domain $\mathcal{M}$ into the following three regimes:
\begin{itemize}
    \item Regime $\text{I}^{(c)}$: $r_{m}=\Omega\left( a(n)^{-1} \right)$
    \item Regime $\text{II}^{(c)}$: $r_{m} = o\left( a(n)^{-1} \right)$ and $\Omega\left( K^{1+\epsilon}\right)$
		\item Regime $\text{III}^{(c)}$: $r_{m} = o\left( K^{1+\epsilon}\right)$ and $\Omega(K)$,
\end{itemize}
\noindent where $\epsilon > 0$ is an arbitrarily small constant. Let $\mathcal{I}^{(c)}_{1}$, $\mathcal{I}^{(c)}_{2}$, and $\mathcal{I}^{(c)}_{3}$ be partitions of $\mathcal{M}$ consisting of content objects belonging to Regimes I\(^{(c)}\),  II\(^{(c)}\),  and III\(^{(c)}\), respectively. Now, characterize the transfer delay for each content $m \in \mathcal{M}$ according to the three regimes to simplify the objective function in \eqref{eq:ofc}. \\
\noindent \textbf{Transfer Delay for Content \(m\in\mathcal{I}^{(c)}_{1}\):}
In Regime I\(^{(c)}\), let $q_m$ be the integer such that $0 \leq q_m \leq K-1$, $(r_m-q_m) a(n) \geq 1$,  and $(r_m-q_m -1)a(n)<1$. Now, the transfer delay for each content $m \in \mathcal{I}^{(c)}_{1}$ is given by
\begin{align*} %\label{eq:davgc1ss}
\sum_{j=0}^{K-1} &\frac{1}{\min\left( 1,(r_m-j)a(n) \right)}\\&= \sum_{j=0}^{q_m}1  +\sum_{j=q_m+1}^{K-1} \frac{1 }{(r_m-j) a(n) }\\& = (q_m+1 ) + \frac{1}{a(n)}\log \left(\frac{r_m-q_m-1}{r_m-K}\right),
\end{align*}
\noindent where the second equality holds due to the harmonic series. By definition of $q_m$, we have  $r_m-q_m-1= \Theta \left( a(n)^{-1}\right)$, which gives us  
\begin{align} \label{eq:dm_r1} 
\sum_{j=0}^{K-1} &\frac{1}{\min\left( 1, (r_m-j) a(n) \right)} \nonumber
\\&= r_m- a(n)^{-1} + a(n)^{-1} \log\left(\!\frac{\!a(n)^{-1}}{r_m - K} \right). 
\end{align}
\noindent Let $z=\!\left(\!r_m\!-K\!-a(n)^{-1}\!\right)/a(n)^{-1}$. Then, it follows that $\log\left(\frac{a(n)^{-1}}{r_m-K}\!\right)=-\log\left(1+z\right)$ and $ \log\left(1+ z\right) = z + O\left(z^2\right)$ due to $z=o(1)$ in Regime I\(^{(c)}\). This finally results in
\begin{equation} \label{eq:davgc1} 
\sum_{j=0}^{K-1} \frac{1}{\min\left( 1, (r_m-j) a(n) \right)} =   \Theta\left(K\right)   \hspace{0.5cm}\textrm{for } m \in \mathcal{I}^{(c)}_{1}.
\end{equation}
\noindent \textbf{Transfer Delay for Content \(m\in\mathcal{I}^{(c)}_{2}\):}
In Regime II\(^{(c)}\), the transfer delay for each content $m \in \mathcal{I}^{(c)}_{2}$ is given by
\begin{align*}
\sum_{j=0}^{K-1} \frac{1}{\min\left( 1, (r_m-j) a(n) \right)} &=  \sum_{j=0}^{K-1} \frac{1}{ (r_m-j) a(n) }  
\\& =  \frac{1}{a(n)} \log\left(\frac{r_m}{r_m-K} \right).
\end{align*}
\noindent Let $z= K/r_m$. Then, it follows that $\log\left(\frac{ r_m }{r_m - K} \right)=-\log\left(1-z\right)$ and $\log\left(1-z\right)= -z + O\left(z^2\right)$ due to $z=o(1)$ in Regime II\(^{(c)}\). This finally results in
\begin{equation} \label{eq:davgc2}
\sum_{j=0}^{K-1}\!\! \frac{1}{\min\left( 1, (r_m-j) a(n) \right)}=  \Theta\left(\frac{K}{a(n)r_m}\!\right) \hspace{0.3cm} \textrm{for }  m \in \mathcal{I}^{(c)}_{2}.
\end{equation}
\noindent \textbf{Transfer Delay for Content \(m\in\mathcal{I}^{(c)}_{3}\):}
 In Regime III\(^{(c)}\), the transfer delay for each content $m \in \mathcal{I}^{(c)}_{3}$ is given by
\begin{align}\label{eq:dm_r3} 
\sum_{j=0}^{K-1} \frac{1}{\min\left( 1, (r_m-j) a(n) \right)} &=  \sum_{j=0}^{K-1} \frac{1 }{(r_m-j) a(n) } \nonumber
\\ &= \frac{1}{a(n)} \left(\! \log r_m-  \log(r_m-K) \!\right).
\end{align}
\noindent In the regime, we have $r_{m} = o\left( K^{1+\epsilon}\right)$ and $\Omega(K)$ for an arbitrarily small $\epsilon>0$. Thus, it follows that $\log r_m -  \log(r_m-K) = \Theta(\log r_m )$ $=$ $\Theta(\log K)$, which results in
\begin{align} \label{eq:davgc3} 
\sum_{j=0}^{K-1} \frac{1}{\min\left( 1, (r_m-j) a(n) \right)} = \Theta\left(\frac{\log K}{a(n)} \right)  \hspace{0.3cm} \textrm{for }  m \in \mathcal{I}^{(c)}_{3}.
\end{align}
Now, using \eqref{eq:davgc1}, \eqref{eq:davgc2}, and \eqref{eq:davgc3}, we can establish the following equivalent optimization problem to the original problem in \eqref{eq:opc}:
\begin{align}\label{eq:opsc}
\min_{\left\{r_{m}\right\}_{m \in \mathcal{M}}}\!\!\left(\!\sum_{m \in \mathcal{I}^{(c)}_{1}}\!\!\!\!Kp^{pop}_{m}\!+\!\!\!\!\sum_{m\in \mathcal{I}^{(c)}_{2}}\!\!\!\frac{p^{pop}_{m}K}{a(n)r_m}\!+\!\!\!\sum_{m \in \mathcal{I}^{(c)}_{3}}\!\!\!\!\frac{p^{pop}_{m}\log K}{a(n)}\!\!\right)
\end{align}
subject to
\begin{align*}
 \sum_{m=1}^Mr_{m} \leq Sn,
\end{align*}
\begin{align*}
 K \leq r_{m} \leq a(n)^{-1}+K.
\end{align*}
\subsubsection{Solving the Simplified Optimization Problem}~\label{section:622}
The Lagrangian function corresponding to \eqref{eq:opsc} is given by 
\begin{align*} \label{eq:LFc}
& \mathcal{L}\left(\left\{r_{m}\right\}_{m \in \mathcal{M}},\delta,\left\{\sigma_m\right\}_{m \in \mathcal{M}}, \left\{\mu_m\right\}_{m \in \mathcal{M}}\right) = \sum_{m \in \mathcal{I}^{(c)}_{1}} K p^{pop}_{m} \\& +\sum_{m \in \mathcal{I}^{(c)}_{2}} \frac{p^{pop}_{m}K}{a(n)r_m}+ \!\!\sum_{m \in \mathcal{I}^{(c)}_{3}}\frac{p^{pop}_{m}\log K}{a(n)}\!\!   + \delta \left(\sum_{m= 1}^{M} r_{m} - Sn \right) \\ & + \sum_{m= 1}^{M} \sigma_m \left( K- r_{m} \right) + \sum_{m= 1}^{M} \mu_m \left( r_{m} - a(n)^{-1}-K\right),
\end{align*}
\noindent where $\mu_m, \sigma_m, \delta \in \mathbb{R}$. The KKT conditions for \eqref{eq:opsc} are then given by
\begin{equation}\label{eq:KKT1_c}
\frac{\partial \mathcal{L}\left(\left\{\hat{r}_{m}\right\}_{m \in \mathcal{M}},\hat{\delta}, \left\{\hat{\mu}_m \right\}_{m \in \mathcal{M}}, \left\{\hat{\sigma}_m\right\}_{m \in \mathcal{M}}\right) }{\partial \hat{r}_{m}}= 0,  
\end{equation}
\begin{equation}\label{eq:KKT2_c}
\hat{\sigma}_m \left( K- \hat{r}_{m} \right) =0 , 
\end{equation}
\begin{equation}\label{eq:KKT6_c}
\hat{\mu}_m \left( \hat{r}_{m} - a(n)^{-1}- K\right) =0 , 
\end{equation}
\begin{equation}\label{eq:KKT3_c}
\hat{\delta} \left(\sum_{m= 1}^{M} \hat{r}_{m} - Sn \right) =0 ,
\end{equation}
\begin{equation*}\label{eq:KKT4_c}
\hat{\delta} \geq 0 ,
\end{equation*}
\begin{equation*}\label{eq:KKT5_c}
\hat{\sigma}_m \geq 0 , 
\end{equation*}
\begin{equation*}\label{eq:KKT7_c}
\hat{\mu}_m \geq 0  
\end{equation*}
\noindent for all $m \in \mathcal{M}$, where $\hat{r}_m, \hat{\delta}$, $\hat{\mu}_m$, and $\hat{\sigma}_m$ represent the optimized values. Let the content indice $m^{(c)}_1\in \mathcal{I}^{(c)}_{2}$ and $m^{(c)}_2\in \mathcal{I}^{(c)}_{3}$ denote the smallest content indice belonging to Regimes II\(^{(c)}\) and III\(^{(c)}\), respectively. In the following, we introduce a lemma that presents an important characteristic of the optimal cache allocation strategy $\left\{\hat{r}_{m}\right\}_{m=1}^M$.
\begin{lemma}\label{le:1_c}
The optimal cache allocation strategy denoted by $\left\{\hat{r}_{m}\right\}_{m=1}^M$ in \eqref{eq:opsc} is non-increasing with $m \in \mathcal{M}$.
\end{lemma}
\begin{proof} 
 Deferred to Appendix \ref{AppendixA_c}.
\end{proof}
Lemma \ref{le:1_c} allows us to establish the second main result regarding the optimal cache allocation strategy for the MDS-coded caching scenario.
\begin{prop}\label{th:oprep_c}
Consider the content-centric mobile ad hoc network model employing subpacketization and following the MDS-coded random reception strategy in Section~{\em\ref{section:322}}. The order-optimal cache allocation strategy is given by
\begin{equation} \label{eq:oprep_c} 
\hat{r}_m = \begin{cases}
    a(n)^{-1} &  m \in \left\{1, \cdots, m^{(c)}_1-1 \right\} \\
    \frac{\sqrt{p^{pop}_m}}{\sum_{\widetilde{m}= m^{(c)}_1}^{m^{(c)}_2-1} \sqrt{p^{pop}_{\widetilde{m}}} } S^{(c)} &  m \in \left\{m^{(c)}_1, \cdots, m^{(c)}_2-1 \right\} \\
		 K &  m \in \left\{m^{(c)}_2, \cdots, M \right\},
\end{cases} 
\end{equation}
where $p^{pop}_m$ is given in \eqref{eq:zipf}, $ S^{(c)}= Sn - (m^{(c)}_1-1)a(n)^{-1}- (M-m^{(c)}_2+1)K $, and the boundaries between any two regimes are defined by content indice $m^{(c)}_1$ and $m^{(c)}_2$, which are given by
\begin{equation} \label{eq:am2_c} 
m^{(c)}_2\! =\!\! \begin{cases}
		\Theta\left(\min\left\{M,(n-M)^{\frac{2}{\alpha}} \right\}\right) &\!\! \alpha > 2 \\
		\Theta\left( \min\left\{M, (n-M)\left(\frac{a(n)^{-1}}{K}\right)^{\frac{2}{\alpha}-1} \right\}\right)  & \!\!\alpha \leq 2
\end{cases} 
\end{equation}
\noindent and 
\begin{equation} \label{eq:am1_c} 
m^{(c)}_1 = \begin{cases}	
		\Theta\left(\min\left\{M,\left(\frac{K(n-M)}{a(n)^{-1}}\right)^{\frac{2}{\alpha}}\!\!\right\}\right) &\!\! \alpha > 2 \\
		\Theta\left( \min\left\{M, (n-M)Ka(n) \right\} \right)  & \alpha \leq 2,
\end{cases} 
\end{equation}
respectively.
\end{prop}
\begin{proof}
 Deferred to Appendix \ref{AppendixD_c}. 
\end{proof}
From Proposition~\ref{th:oprep_c}, the order-optimal cache allocation strategy is partitioned into three parts, and the content indice $m^{(c)}_1$ and $m^{(c)}_2$  are specified as a function of key parameters $K$, $M$, $a(n)$, and $\alpha$. Similarly as in the uncoded caching scenario, our MDS-coded cache allocation strategy under the total caching constraint in \eqref{eq:opc} can be extended to satisfy the local caching constraints when the replica allocation algorithm in Section~\ref{section:52} is employed in which for each content $m \in \mathcal{M}$, $\left\lceil \hat{r}_m\right\rceil$ MDS-coded subpackets are cached instead of $\left\lceil \hat{X}_m\right\rceil$ replicas. Based on the same argument as those in Remark~\ref{R:1}, the local cache size constraints hold within a factor of 2.

In the next subsection, we characterize the optimized minimum average content transfer delay $D_{avg}(n)$ by adopting the order-optimal cache allocation strategy presented in Proposition~\ref{th:oprep_c} and also analyze the impact of key parameters $K$, $M$, $a(n)$, and $\alpha$ on the order-optimal performance. 
\subsection{Order-Optimal Performance}~\label{section:63}
In this subsection, we compute the minimum average content transfer delay $D_{avg}(n)$ using the order-optimal cache allocation strategy obtained in Proposition~\ref{th:oprep_c}.
\begin{theorem}\label{th:delay_c} Consider a content-centric mobile ad hoc network model with subpacketization adopting the optimal cache allocation strategy $\left\{\hat{r}_{m}\right\}_{m=1}^M$ in~\eqref{eq:oprep_c} and following the MDS-coded random reception strategy. Then, the minimum average content transfer delay $D_{avg}(n)$ is given by
\smallskip
\begin{align}
D_{avg}(\!n\!)\!=\!\!\!\left\{\begin{array}{lll} \!\!\!\!\Theta\!\!\left(K\right) 
&\!\!\!\textrm{\!\!$m^{(c)}_1\!\!=\!\Theta(\!M\!)$} \\
\!\!\!\!\Theta\!\!\left(\!\!\max \!\!\left\{\!K,\frac{\left(H_\frac{\alpha}{2}(M)\right)^2}{H_\alpha(M) na(n)}\right\}\!\!\right) &\!\!\!\textrm{\!\!$m^{(c)}_2\!\!=\!\Theta(\!M\!)$ \upshape} \\
\hspace{4.8cm}\textrm{\upshape and} & \!\!\!\textrm{\!\!$m^{(c)}_1\!\!=\!o(\!M\!)$} \\
\!\!\!\!\Theta\!\left(\!\!\max\!\!\left\{\!\! K,\frac{a(n)^{-1}\left(H_\frac{\alpha}{2}(m^{(c)}_2)\right)^2}{H_\alpha(M) (n-M)},\!\frac{\log K}{a(n)}\!\right\}\!\!\right)
&\!\!\!\!\textrm{\!\!$m^{(c)}_2\!\!=\!o(\!M\!)$},
\end{array}\right. \label{eq:delay_c2} 
\end{align}
\noindent where $K$ is the number subpackets of each content, $a(n)$ is the area in which a node can communicate with other nodes, and $H_\alpha(M)$ and $H_\frac{\alpha}{2}(M)$ are given in \eqref{eq:H}  and \eqref{eq:H2}, respectively.
\end{theorem}
\smallskip
\begin{proof}
 Deferred to Appendix \ref{AppendixE_c}.
\end{proof}
From Theorems~\ref{th:TD} and~\ref{th:delay_c}, the maximum achievable per-node throughput $\lambda(n)$ is given by
\begin{equation}  \label{eq:pnt_c}
\lambda(n)\!\!=\!\!\left\{\begin{array}{lll} \!\!\!\!\Theta\!\!\left(\frac{1}{na(n)K}\!\!\right)
&\!\!\!\!\textrm{$m^{(c)}_1\!\!=\!\Theta(\!M\!)$} \\
\!\!\!\!\Theta\left(\!\min \left\{\frac{1}{na(n)K},\frac{H_\alpha(M)}{\left(H_\frac{\alpha}{2}(M)\right)^2}\right\}\!\!\right) &\!\!\!\!\textrm{$m^{(c)}_2\!\!=\!\Theta(\!M\!)$} \\
%&\!\!\!\textrm{and} \\
\hspace{5cm}\textrm{\upshape and} &\!\!\!\!\textrm{$m^{(c)}_1\!\!=\!o(\!M\!)$} \\
\!\!\!\! \Theta\!\!\left(\!\!\min\left\{\!\!\frac{1}{na(n)K},\!\frac{H_\alpha(M) (n-M)}{n\left(H_\frac{\alpha}{2}(m^{(c)}_2)\right)^2},\!\frac{1}{n\log K}\right\}\!\!\right)\!\!
&\!\!\!\!\textrm{$m^{(c)}_2\!\!=\!o(\!M\!)$}.
\end{array}\right.
\end{equation}
Similarly as in the uncoded caching case, the average content transfer delay $D_{avg}(n)$ and the per-node throughput $\lambda(n)$ for MDS-coded caching scale with respect to $a(n)$, $K$, $M$ and $\alpha$. To validate the analytical results obtained in Sections~\ref{section:5} and~\ref{section:6}, we perform intensive numerical evaluation in the next section. 

\section{Numerical Evaluation and Performance Comparison}~\label{section:7} 
In this section, we perform intensive computer simulations with finite system parameters $a(n)$, $K$, $M$, and $\alpha$ to obtain the numerical solutions to the optimization problems in \eqref{eq:op_s} and \eqref{eq:opsc}. We compare the numerical evaluations with the analytical results presented in Sections~\ref{section:5} and~\ref{section:6} to validate our analysis. We first validate the order-optimal caching allocation strategies presented in \eqref{eq:oprep_s} and \eqref{eq:oprep_c} and highlight the impact of system parameters according to the operating regimes. Then, we compare the order-optimal performance  on the average content transfer delay $D_{avg}(n)$ for uncoded and MDS-coded caching scenarios.

\subsection{Order-Optimal Cache Allocation Strategy}~\label{section:71}
Figure~\ref{fig:caching_uncoded} is an illustration of the optimal caching strategy for the uncoded caching case employing sequential reception. We can observe the consistency between the analytical results in Fig.~\ref{fig:Aseq_rep} obtained using Proposition \ref{th:oprep_s} and the results obtained by numerically solving the problem in~\eqref{eq:op_s} in Fig.~\ref{fig:Nseq_rep} for $M=250$, $K=20$, and $n=30000$. We can see how the optimal number of replicas $\hat{X}_m$ behaves according to different values of the area $a(n)$ and the Zipf exponent $\alpha$ (i.e., values corresponding to their respective operating regime) as depicted in Fig.~\ref{fig:caching_uncoded}.
When $\alpha=0.5$, the boundary between Regimes I\(^{(u)}\) and II\(^{(u)}\) is given by $m^{(u)}_1=\Theta\left(\frac{(na(n))^4}{M^3}\right)$. In this case, if $a(n)= \Theta\left(\log n/n\right)$, then the optimal number of replicas $\hat{X}_m$ is monotonically decreasing with a slope of $\alpha/2$, i.e., the caching strategy operates in Regime II\(^{(u)}\). When we increase $a(n)$ $($e.g., $a(n)= \Theta\left(M^{0.8}/n\right))$, the caching strategy operates in both Regimes I\(^{(u)}\) and II\(^{(u)}\). On the other hand, when $\alpha=2$, the boundary between two regimes is given by $m^{(u)}_1=\Theta\left(\frac{na(n)}{\log M}\right)$. In this case, the range of Regime I\(^{(u)}\) tends to be wider than the case of $\alpha=0.5$, as shown in Fig~\ref{fig:caching_uncoded}. 
\begin{figure}[t]
\centering
\subfigure[Analytical  results]{\includegraphics[width=0.48\linewidth]{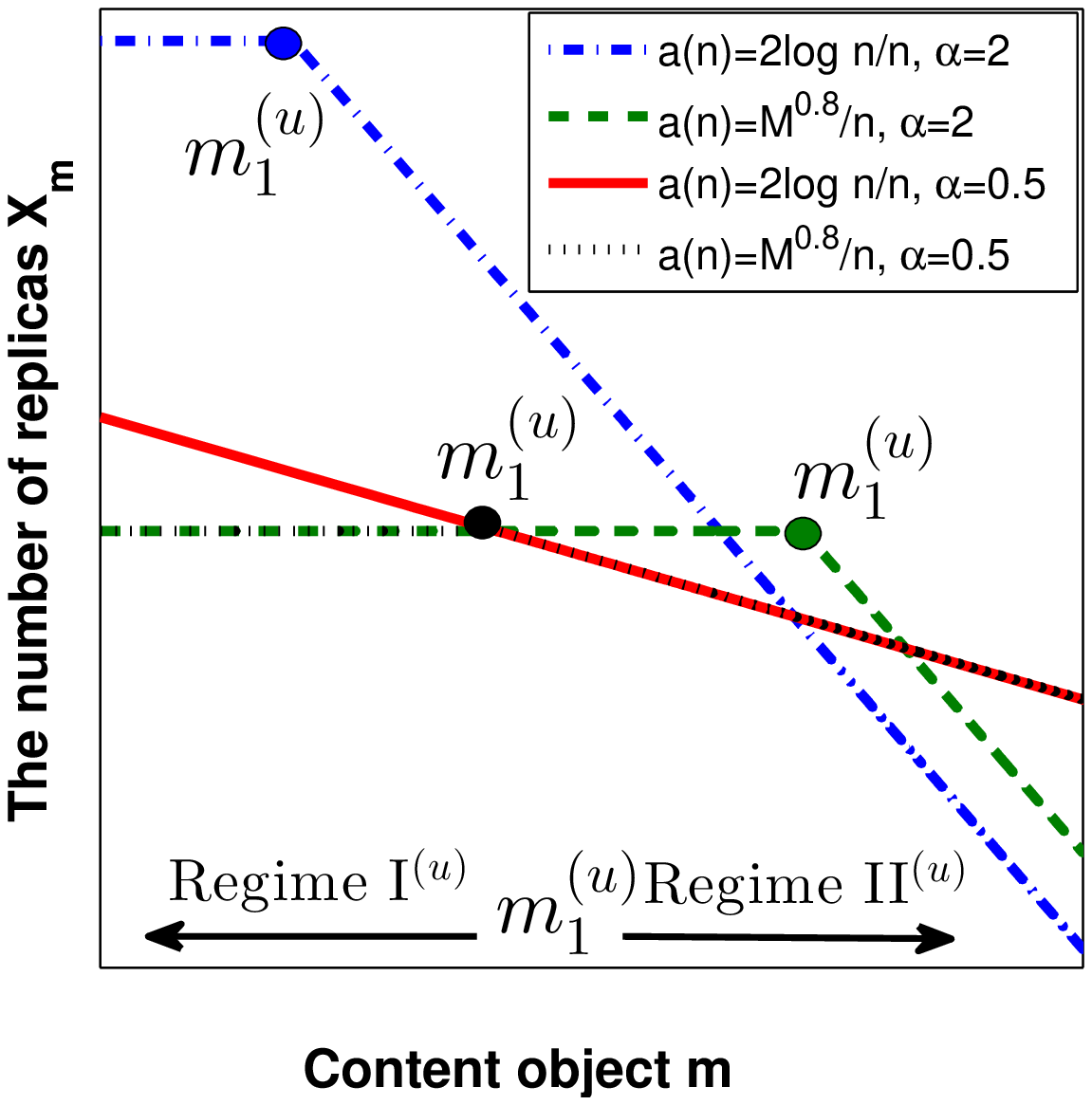} \label{fig:Aseq_rep}}
\subfigure[Numerical  results]{\includegraphics[width=0.48\linewidth]{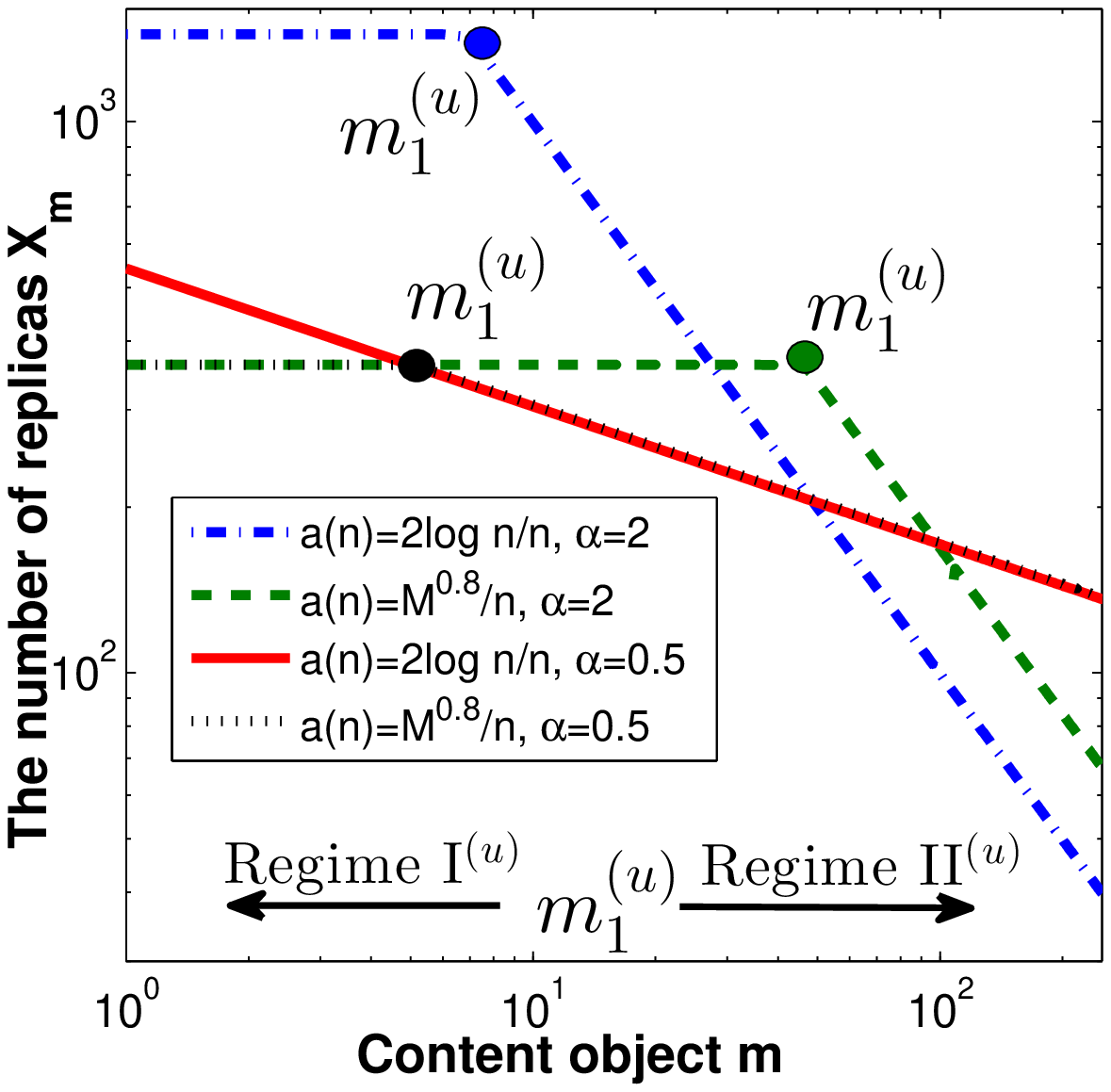} \label{fig:Nseq_rep}}
\caption{ Optimal cache allocation strategy  versus content object $m$ for the uncoded caching case employing sequential reception.}
\label{fig:caching_uncoded}
\end{figure}

\begin{figure}[t]
\centering
\subfigure[Analytical results]{\includegraphics[width=0.48\linewidth]{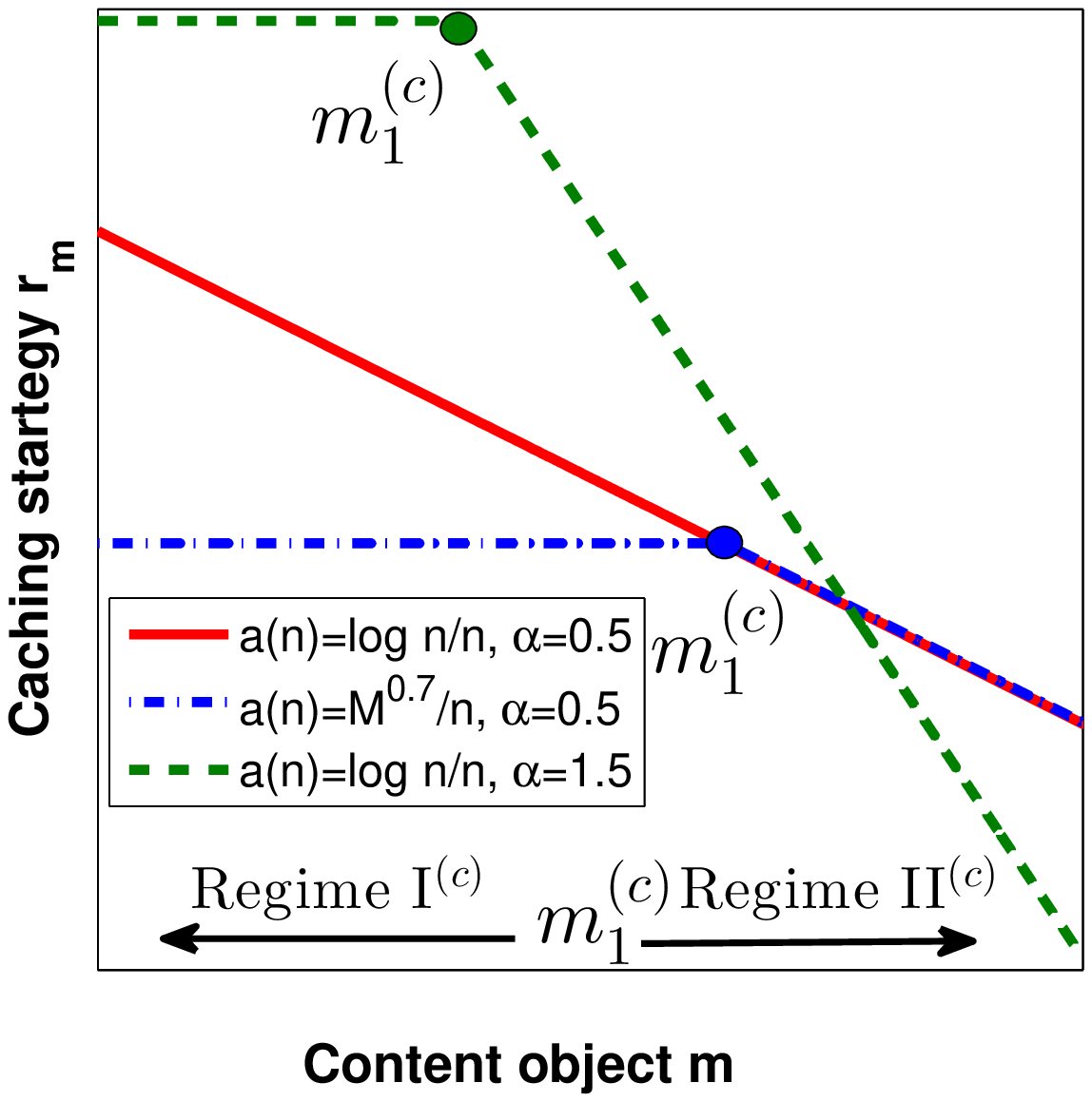} \label{fig:Aran_rep}}
\subfigure[Numerical results]{\includegraphics[width=0.48\linewidth]{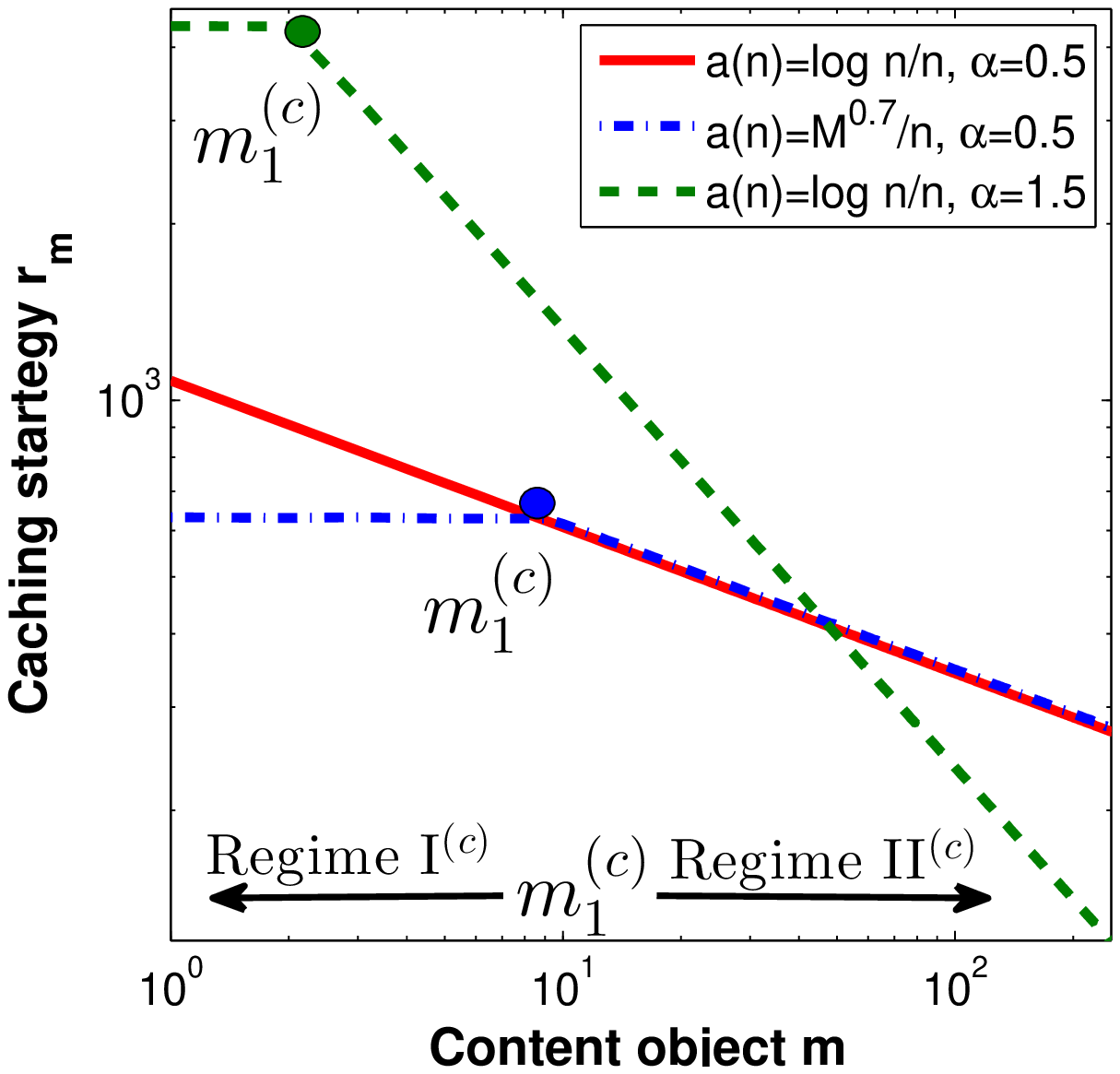} \label{fig:Nran_rep}}
\caption{ Optimal cache allocation strategy versus content object $m$ for the MDS-coded caching case employing random reception.}
\label{fig:caching_coded}
\end{figure}

In Fig.~\ref{fig:caching_coded}, the optimal caching strategy for the MDS-coded caching case employing random reception is illustrated, where the analytical results are depicted in Fig.~\ref{fig:Aran_rep} obtained by Proposition \ref{th:oprep_c}. The results obtained by numerically solving the problem in~\eqref{eq:opsc} are also shown in Fig.~\ref{fig:Nran_rep} for $M=250$, $K=3$, and $n=30000$. Similarly as in the uncoded caching case, we can see how the optimal number of MDS-coded subpackets $\hat{r}_m$ behaves according to different values of the area $a(n)$ and the Zipf exponent $\alpha$ (i.e, values corresponding to their respective operating regime) as depicted in Fig.~\ref{fig:caching_coded}.
Form Propositions \ref{th:oprep_s} and \ref{th:oprep_c}, an important observation is that for given system parameters, the range of Regime I\(^{(c)}\) (the MDS-coded caching case) tends to scale $K$ times wider than that of Regime I\(^{(u)}\) (the uncoded caching case).
\subsection{Order-Optimal Performance}~\label{section:72}
In Fig.~\ref{fig:del_area}, we illustrate how the optimal average content transfer delay $D_{avg}(n)$ behaves according to different values of the area $a(n)$ and the Zipf exponent $\alpha$. We can observe the consistency between the analytical results in Fig.~\ref{fig:Adel_area} obtained using Theorems~\ref{th:delay_s} and~\ref{th:delay_c} and the results obtained by numerically solving the problems in~\eqref{eq:op_s} and~\eqref{eq:opsc} in Fig.~\ref{fig:Ndel_area}, respectively, for $M=250$, $K=20$, and $n=30,000$. When $\alpha=3$, the average content transfer delay of $D_{avg}(n)= \Theta(K)$ is achieved for both the uncoded and MDS-coded caching cases, which is the minimum that we can hope for as far as $a(n)=\Omega(\log n/n)$. The performance difference between the uncoded and the MDS-coded caching scenarios becomes prominent when $\alpha < 2$ as shown in the figure. From the fact that for $\alpha=1.5$, the average delay $D_{avg}(n)$ is given by $\Theta\left(\max\left\{K, \frac{KM^{0.5}}{na(n)}\right\}\right)$ and $\Theta\left(\max\left\{K,\frac{M^{0.5}}{na(n)}\right\}\right)$ for the uncoded and the MDS-coded caching cases, respectively. Moreover, we have $D_{avg}(n) =\Theta(K)$ when $a(n)$ scales as $\Omega\left(\sqrt{M}/n\right)$ and as $\Omega\left(\sqrt{M}/nK\right)$ for the uncoded and MDS-coded caching cases, respectively. Similarly, for $\alpha=0.5$, it follows that $D_{avg}(n) =\Theta(K)$ when $a(n)$ scales as $\Omega(M/n)$ and as $\Omega(M/nK)$ for the uncoded and MDS-coded caching cases, respectively. For $\alpha<2$, based on the above arguments and from Theorem~\ref{th:TD}, the per-node throughput $\lambda(n)$ for the MDS-coded caching case scales $K$ times larger than the uncoded caching case while attaining the order-optimal delay $D_{avg}(n) =\Theta(K)$.
\begin{figure}[t]
\centering
\subfigure[Analytical results]{\includegraphics[width=0.48\linewidth]{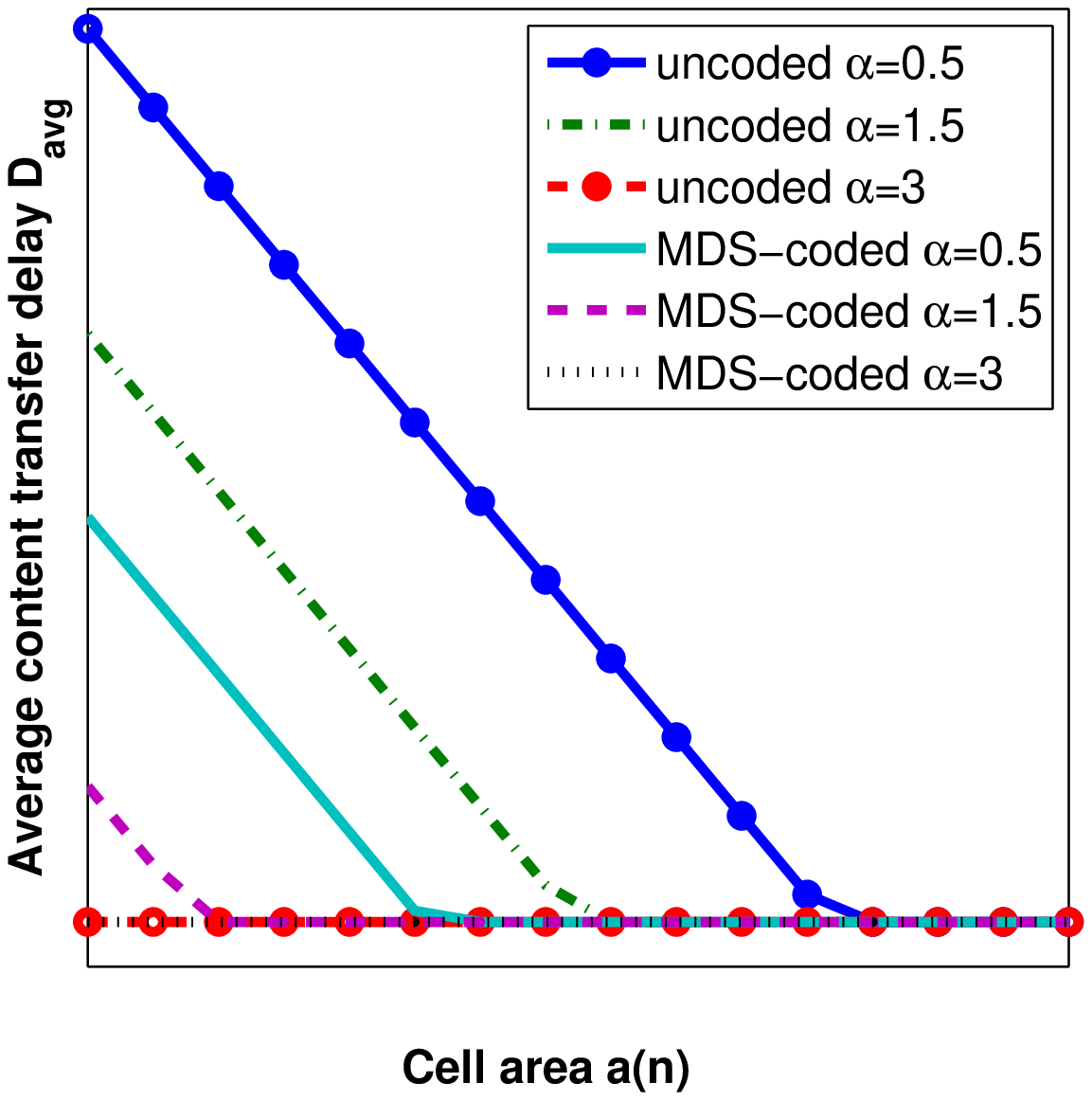} \label{fig:Adel_area}}
\subfigure[Numerical results]{\includegraphics[width=0.48\linewidth]{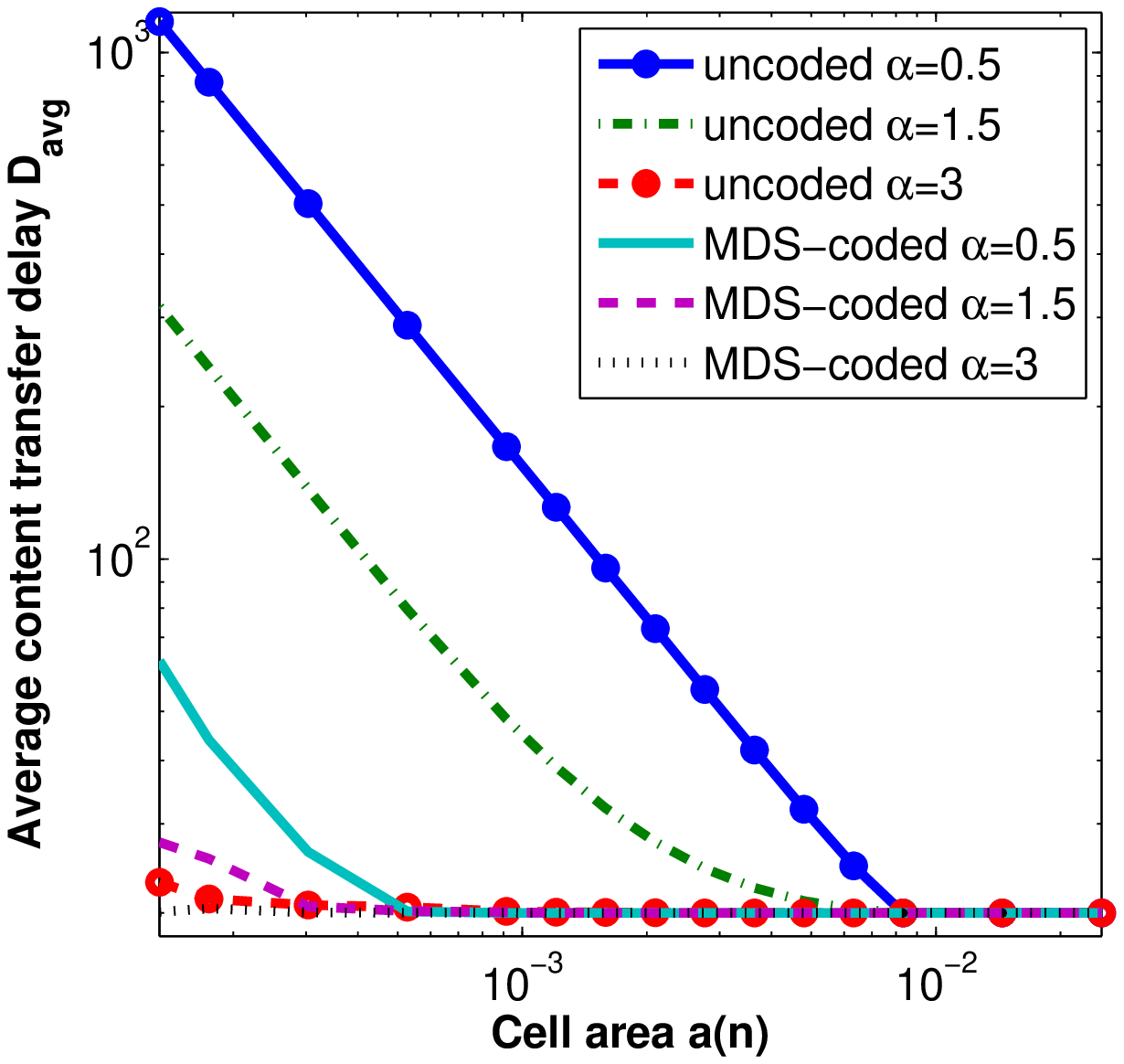} \label{fig:Ndel_area}}
\caption{The average content transfer delay $D_{avg}(n)$ versus cell area $a(n)$.}
\label{fig:del_area}
\end{figure}
\begin{figure}[t]
\centering
\subfigure[Analytical results]{\includegraphics[width=0.48\linewidth]{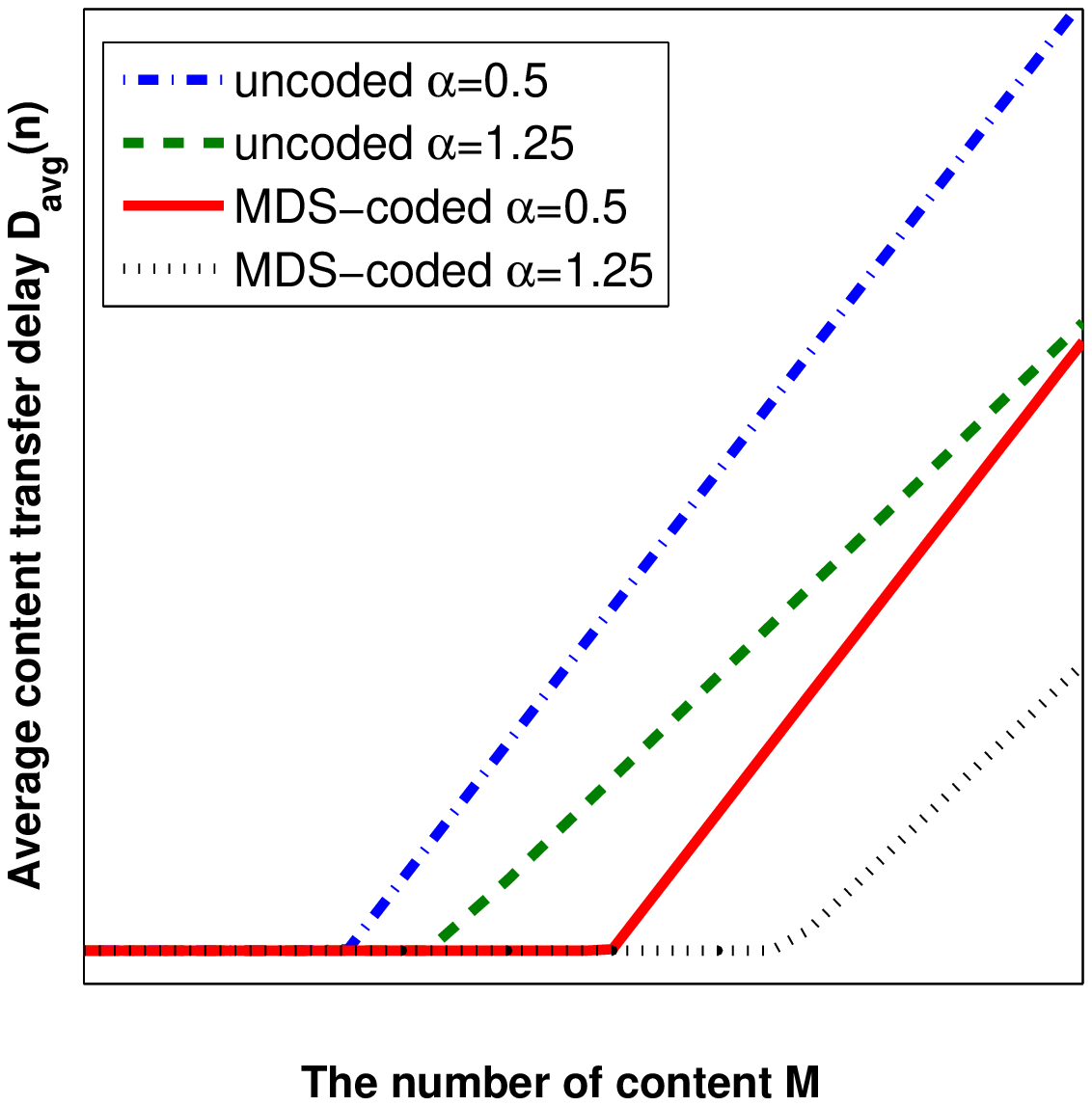} \label{fig:Adel_content}}
\subfigure[Numerical results]{\includegraphics[width=0.48\linewidth]{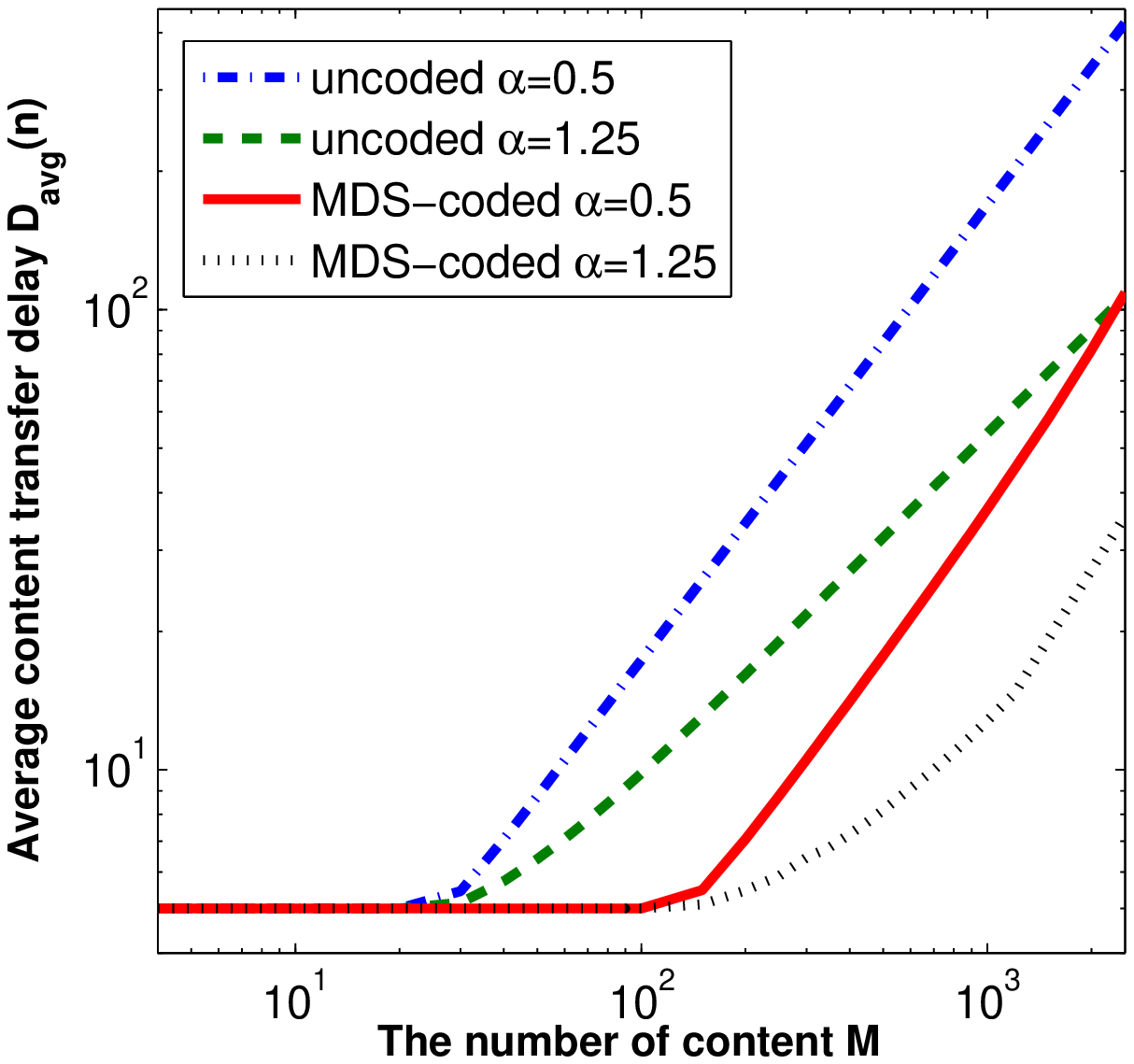} \label{fig:Ndel_content}}
\caption{The average content transfer delay $D_{avg}(n)$ versus the number of content $M$.}
\label{fig:del_content}
\end{figure}

Figure~\ref{fig:del_content} is an illustration of how the optimal average content transfer delay $D_{avg}(n)$ behaves according to different values of the number of content $M$ and the Zipf exponent $\alpha$. Figure~\ref{fig:Adel_content} illustrates the analytical results obtained from Theorems~\ref{th:delay_s} and~\ref{th:delay_c}, and Fig.~\ref{fig:Ndel_content} illustrates the results obtained by numerically solving the problems in~\eqref{eq:op_s} and~\eqref{eq:opsc} for $K=5$, $a(n)= n/\log n$, and $n=30,000$. The performance difference between the uncoded and the MDS-coded caching cases also becomes prominent when $\alpha < 2$. For $\alpha=1.25$, the average delay $D_{avg}(n)$ is given by $\Theta\left(\max\left\{K, \frac{KM^{0.75}}{\log n}\right\}\right)$ and $\Theta\left(\max\left\{K,\frac{M^{0.75}}{\log n}\right\}\right)$ for the uncoded and the MDS-coded caching cases, respectively. Moreover, we have $D_{avg}(n) =\Theta(K)$ when $M$ scales as $O\left(\left(\log n\right)^\frac{4}{3}\right)$ and as $O\left(\left(K \log n\right)^\frac{4}{3}\right)$ for the uncoded and the MDS-coded caching cases, respectively. Similarly, for $\alpha=0.5$, it follows that $D_{avg}(n) =\Theta(K)$ when $M$ scales as $O(\log n)$ and as $O(K \log n)$ for the uncoded and the MDS-coded caching cases, respectively.

\section{Extension to the Random Walk Mobility Model}~\label{section:8}
In this section, we extend our study to another scenario where each node moves independently according to the random walk mobility model studied in~\cite{alfano, rw, algamal}. In the mobility model, the position $d(t)$ of a node at time slot $t$ is updated by $d(t)= d(t-1)+ y_t$, where $y_t$ is a sequence of independent and identically distributed (i.i.d.) random variables that represent a node's flight vector with an average flight length $L= \mathbb{E}\left[\left\|y_t\right\|\right]$. Here, $L$ is assumed to scale as the transmission range $R$ (i.e., $R=\Theta(L)$) as in~\cite{rw, algamal}\footnote{Note that when the transmission range $R$ scales slower than the flight length $L$ (i.e., $R=o(L)$), one can achieve the same results as those for $R=\Theta(L)$ based on~\cite[Lemma 7]{alfano}. The single-hop scenario is known to be appropriate for $R=o(L)$, where higher throughput can be achieved compared to the case using multihop relaying protocols~\cite[Section IV]{alfano}.}.  Likewise, we adopt the single-hop-based content delivery that does not employ any relaying strategies. Thus, a requesting node can successfully retrieve its desired content only if a potential source node is within the transmission range $R= \Theta(L)$ in a given time slot. Otherwise, it moves until it finds an available source node for the request. The content delivery protocol and reception strategies essentially follow the same line as those in Section~\ref{section:3}. We state the following lemma introduced in~\cite{alfano} in terms of our notations for completeness.
\begin{lemma}[\!\protect{\cite[Lemma 7]{alfano}}]\label{le:tavg} Consider two arbitrary nodes that are uniformly distributed over a region of unit area at time $t=0$ and assume that each node moves independently according to the random walk model with average flight length $L$. Then, the average first hitting time $T_{avg}$ required such that the distance between the nodes is less than or equal to $R=\Theta(L)$ is given by
\begin{align} \label{eq:Tavg}
T_{avg} = O\left(\frac{\log n}{R^2} \right) \textrm{\upshape and }\Omega\left(\frac{1}{R^2}\right).   
\end{align} 
\end{lemma}
Note that both upper and lower bounds on the average first hitting time $T_{avg}$ are of the same order within a factor of $\log n$. From Lemma \ref{le:tavg}, we establish the following lemma, which formulates the average content transfer delay $D_{avg}(n)$ for both the uncoded sequential reception in Section~\ref{section:321} and the MDS-coded random reception in Section~\ref{section:322} under the random walk mobility model.
\begin{lemma}\label{le:otdRW} 
Consider a content-centric mobile network in which each node moves according to the random walk mobility model with average flight length $L$ and retrieves its requests according to the content delivery protocol in Section~\ref{section:31}. The average content transfer delay $D_{avg}(n)$ for the uncoded caching case employing the sequential reception strategy in Section~\ref{section:321} is given by
\begin{align} \label{eq:davg_SRW} 
D_{avg}(n) &\hspace{0.13cm}= O\left(\sum\limits_{m=1}^{M} \frac{Kp^{pop}_m}{\min\left(1, \frac{a(n)X_{m}}{\log n}\right)} \right) \nonumber
\\ & \textrm{\upshape and } \Omega\left(\sum\limits_{m=1}^{M} \frac{Kp^{pop}_m}{\min\left(1, a(n)X_{m}\right)}\right)
\end{align}
\noindent and $D_{avg}(n)$ for the MDS-coded caching case employing the random reception strategy in Section~\ref{section:322} is given by
\begin{align} \label{eq:davg_CRW} 
D_{avg}(n) &\hspace{0.13cm}=  O\left(\sum\limits_{m=1}^M \sum\limits_{j=0}^{K-1} \frac{p^{pop}_m}{\min\left( 1, \frac{(r_m-j)a(n)}{\log n}\right)} \right) \nonumber \\ &  \textrm{\upshape and } \Omega\left(\sum\limits_{m=1}^M  \sum\limits_{j=0}^{K-1} \frac{p^{pop}_m}{\min\left( 1, (r_m-j)a(n)\right)}\right).  
\end{align}
\end{lemma}
\begin{proof} We first consider the uncoded caching case employing the sequential reception strategy. Let $p^{seq}_{m,i}$ be the contact probability that a node requesting the $i$th subpacket of content $m \in \mathcal{M}$ falls within distance $R$ of a node holding the requested subpacket  in a given time slot. Then, by employing the cache allocation strategy $\left\{X_{m}\right\}_{m=1}^M$ and using Lemma \ref{le:tavg}, the contact probability $p^{seq}_{m,i}$ is order-equivalent to $\min\left( 1, \frac{X_{m}}{T_{avg}} \right)$. Then, the number of time slots required to successfully receive content object $m\in \mathcal{M}$ which consists of $K$ subpackets is given by $\Theta \left(\frac{K}{\min\left( 1, \frac{X_{m}}{T_{avg}} \right)} \right)$. Thus, using~\eqref{eq:Tavg}, the average content transfer delay $D_{avg}(n)$ for the content-centric mobile network employing the uncoded sequential reception strategy and following the random walk mobility model is given by~\eqref{eq:davg_SRW}. \\
Next, we characterize the average content transfer delay $D_{avg}(n)$ for the MDS-coded caching case employing the random reception strategy. Let $p^{ran}_{m,j}$ be the contact probability that a node having pending requests for $K-j$ MDS-coded subpackets of content $m$ falls within distance $R$ of a node in a given time slot holding one of the requested MDS-coded subpackets while the requesting node is assumed to have already received $j$ MDS-coded subpackets. Then, by employing the cache allocation strategy $\left\{r_{m}\right\}_{m=1}^M$ and from Lemma~\ref{le:tavg}, the contact probability $p^{ran}_{m,j}$ is order-equivalent to $\min\left( 1, \frac{r_m-j}{T_{avg}} \right)$. Furthermore, the expected number of time slots required to successfully receives content object $m\in \mathcal{M}$ consisting of $K$ MDS-coded subpackets is given by $\Theta \left( \sum_{j=0}^{K-1} \frac{1}{\min\left( 1, \frac{r_m-j}{T_{avg}} \right)} \right)$. Thus, using~\eqref{eq:Tavg}, the $D_{avg}(n)$ for the MDS-coded caching case employing the random reception strategy and following the random walk mobility model is given by~\eqref{eq:davg_CRW}. This completes the proof of the lemma.
\end{proof} 
Under the random walk mobility model, we now turn to analyzing the main results. By comparing Lemmas~\ref{le:otd} and \ref{le:otdRW}, we observe that the average content transfer delay $D_{avg}(n)$ of the random walk mobility model scales as that of the reshuffling mobility model within a factor of $\log n$. Hence, it is straightforward to achieve essentially the same optimal cache allocation strategies and order-optimal throughput--delay trade-offs for both uncoded and MDS-coded caching scenarios as those in the reshuffling mobility model within a polylogarithmic factor (refer to~\eqref{eq:oprep_s},~\eqref{eq:delay_s2},~\eqref{eq:pnt_s},~\eqref{eq:oprep_c},~\eqref{eq:delay_c2}, and~\eqref{eq:pnt_c} for comparison). This implies that as long as the single-hop-based content delivery protocol is adopted, the random walk mobility model does not fundamentally change the results attained from the reshuffling mobility model.

\section{Concluding Remarks}
This paper investigated the utility of subpacketization in a content-centric mobile ad hoc network, where each mobile node equipping finite-size cache space moves according to the reshuffling mobility model and only a subpacket of a content object consisting of $K$ subpackets can be delivered during one time slot due to the fast mobility condition. The fundamental trade-offs between throughput and delay under our network model were first established by adopting single-hop-based content delivery. Order-optimal caching strategies in terms of throughput--delay trade-offs were then presented for both the sequential reception strategy for uncoded caching and the random reception strategy for MDS-coded caching. In addition, our analytical results were comprehensively validated by numerical evalaution. In consequence, it was found that as $\alpha<2$, the  MDS-coded caching strategy has a significant performance gain over the uncoded caching case. More precisely, it was shown that the per-node throughput for MDS-coded caching scales $K$ times faster than that of uncoded caching when the delay is fixed to a minimum. Moreover, for the MDS-coded caching strategy, if $K$ scales faster than $M$, then the best performance on the throughput and delay is achieved. It was also investigated that adopting the random walk mobility model does not essentially change our main results.

An interesting direction for further research is to characterize the optimal throughput--delay trade-off in mobile hybrid networks employing subpacketization, where both mobile nodes and static helper nodes are able to cache a subset of content objects with different capabilities. Potential avenues of another future research in this area include analyzing the optimal throughput--delay trade-off by adopting mobility models that better reflect human mobility patterns in outdoor settings (e.g., the random waypoint mobility and Levy walk mobilty models).

\appendix
%\section{Appendix}
\renewcommand\theequation{\Alph{section}.\arabic{equation}}
\setcounter{equation}{0}
\subsection{Proof of Lemma \ref{le:1_s}} \label{AppendixA_s}
First, note that the optimal cache allocation strategy for $\hat{X}_m$, $m\in\Ic^{(u)}_1$ is trivially non-increasing in Regime I\(^{(u)}\). Thus, in the following, we focus only on Regime II\(^{(u)}\). From the stationary condition in~\eqref{eq:KKT1_s}, we have   
\begin{align} \label{eq:KKT1r2_s} 
- \frac{Kp^{pop}_{m}}{a(n) (\hat{X}_{m})^2} + \hat{\delta} K  + \hat{\sigma}_m =0  \hspace{1cm} m \in \mathcal{I}^{(u)}_{2}.
\end{align} 
From condition~\eqref{eq:KKT2_s}, we have $\hat{\sigma}_m=0$ for $m \in \mathcal{I}^{(u)}_{2}$. Using \eqref{eq:KKT1r2_s}, we have
\begin{align} 
 \hat{\delta} &= \frac{p^{pop}_{m}}{ a(n) \hat{X}^{2}_{m}}, \label{eq:repir22_s} \\
 \hat{X}_{m} &= \sqrt{\frac{p^{pop}_{m}}{ a(n) \hat{\delta}}}, \label{eq:repir2_s} \\
 \hat{\delta}^{\frac{ 1}{2}} &= \frac{\sqrt{p^{pop}_{m}}}{ a(n)^{\frac{1}{2}} \hat{X}_{m} } =   \frac{\sum _{\widetilde{m} \in \mathcal{I}^{(u)}_{2}} \sqrt{p^{pop}_{\widetilde{m}}}}{ \sqrt{a(n)} \sum _{\widetilde{m} \in \mathcal{I}^{(u)}_{2}} \hat{X}_{\widetilde{m}} }. \label{eq:KKT1r2s_s} 
\end{align} 
\noindent By combining \eqref{eq:repir2_s} and \eqref{eq:KKT1r2s_s} we have
\begin{align} \label{eq:opxr2_s} 
\hat{X}_{m} = \frac{\sqrt{p^{pop}_m}}{\sum_{\widetilde{m} \in \mathcal{I}^{(u)}_2} \sqrt{p^{pop}_{\widetilde{m}}} } \sum_{\widetilde{m} \in \mathcal{I}^{(u)}_{2} } \hat{X}_{\widetilde{m}}, \qquad m \in \mathcal{I}^{(u)}_{2}.
\end{align}
Hence, the optimal cache allocation strategy $\left\{\hat{X}_{m}\right\}_{m \in \mathcal{I}^{(u)}_{2}}$ is non-increasing in Regime II\(^{(u)}\).
Now, we are ready to finalize the proof of Lemma \ref{le:1_s}. Consider any content object $j \in \mathcal{I}^{(u)}_{1}$. Then, using~\eqref{eq:KKT1_s} and the fact that $\hat{\sigma}_m = 0$ for $m \in \mathcal{I}^{(u)}_{2}$, we have $\hat{\delta} = \frac{p^{pop}_{j}}{ a(n) \hat{X}^{2}_{j}} -\frac{\hat{\sigma}_j}{K} = \frac{p^{pop}_{m_1^{(u)}}}{ a(n) \hat{X}^{2}_{m^{(u)}_1}}   > 0$. Since  $\hat{X}_{j}= a(n)^{-1}$, $\hat{X}_{m^{(u)}_1}< a(n)^{-1}$, and $\hat{\sigma}_j \geq 0$, we obtain $p^{pop}_j > p^{pop}_{m^{(u)}_1}$, thus resulting in $j < m^{(u)}_1$ due to the feature of a Zipf popularity in \eqref{eq:zipf}. This completes the proof of the Lemma \ref{le:1_s}.
\subsection{Proof of Proposition~\ref{th:oprep_s}}\label{AppendixD_s}
Let us first characterize the content index $m^{(u)}_1$, which identifies the boundary between Regimes I\(^{(u)}\)  and II\(^{(u)}\).  
Since $m^{(u)}_1 \in \mathcal{I}^{(u)}_{2}$ is the smallest index such that $\hat{X}_{m^{(u)}_1} < a(n)^{-1}$, using \eqref{eq:opxr2_s} yields
\begin{align} \label{eq:m1a1_s}
\left(m^{(u)}_1\right)^{-\frac{\alpha}{2}}  < a(n)^{-1} \frac{\sum_{\widetilde{m} = m^{(u)}_1}^{M} \widetilde{m}^{-\frac{\alpha}{2}}}{\sum_{\widetilde{m}= m^{(u)}_1 }^{M} \hat{X}_{\widetilde{m}}}.
\end{align}
\noindent Now, if $m^{(u)}_1 > 1$, then attempting to decrease the index $m^{(u)}_1$ by one and using \eqref{eq:opxr2_s} would result in 
\begin{align} \label{eq:m1a2_s}
\left(m^{(u)}_1-1 \right)^{-\frac{\alpha}{2}}  \geq a(n)^{-1} \frac{\sum_{\widetilde{m} = m^{(u)}_1}^{M} \widetilde{m}^{-\frac{\alpha}{2}}}{\sum_{\widetilde{m}= m^{(u)}_1 }^{M} \hat{X}_{\widetilde{m}}}.
\end{align}
\noindent From \eqref{eq:m1a1_s}, \eqref{eq:m1a2_s}, and $\sum_{\widetilde{m} = m^{(u)}_1 }^{M} K\hat{X}_{\widetilde{m}}= Sn - \sum_{\widetilde{m}=1}^{m^{(u)}_1-1}K \hat{X}_{\widetilde{m}} $ such that the condition \eqref{eq:KKT3_s} is fulfilled, we obtain  
\begin{align}\label{eq:m1h_s} 
\left(m^{(u)}_1 -1\right)^{\frac{\alpha}{2}}  = \Theta \left(  \frac{na(n) - (m^{(u)}_1-1)}{ H_{\frac{\alpha}{2}}(M) - H_{\frac{\alpha}{2}}\left(m^{(u)}_1 -1\right) } \right), 
\end{align}
\noindent where $H_{\frac{\alpha}{2}}(M)$ is given in \eqref{eq:H2}. Using \eqref{eq:m1h_s} and \eqref{eq:H2}, content index $m^{(u)}_1$ is specified as follows: 
\begin{enumerate}
\item For $\alpha >2 $: 
\begin{align} \label{eq:m1_a1} 
m^{(u)}_1  = \Theta \left( \left(na(n) \right)^{\frac{2}{\alpha}} \right);
\end{align}
\item For $\alpha \leq 2 $:
\begin{align}\label{eq:m1_a2} 
m^{(u)}_1  =  \begin{cases}
\Theta \left(M \right) & $if$~a(n)=\Omega\left(\frac{M}{n}\right) \\
\Theta \left(\left(\frac{na(n)}{H_{\frac{\alpha}{2}}(M) }\right)^{\frac{2}{\alpha}}\right) & $otherwise$.
\end{cases}
\end{align}
\end{enumerate}
Now, from Lemma \ref{le:1_s}, the optimal cache allocation strategy is given in \eqref{eq:oprep_s} since $\sum_{\widetilde{m} = m^{(u)}_1 }^{M} \hat{X}_{\widetilde{m}}= n - (m^{(u)}_1-1) a(n)^{-1} $. This completes the proof of Proposition~\ref{th:oprep_s}.
%%\section{Proof of Theorem \ref{th:delay_s}}\label{AppendixE_s}
\subsection{Proof of Theorem \ref{th:delay_s}}\label{AppendixE_s}
In Regimes I$^{(u)}$ and II$^{(u)}$, the minimum $D_{avg}(n)$ resulting from the optimal cache allocation strategy $\left\{\hat{X}_{m}\right\}_{m=1}^M$ in \eqref{eq:oprep_s} is given by
\begin{align*}
D_{avg}(n) =  \sum_{m=1}^{m^{(u)}_1-1} K p^{pop}_{m} + \sum_{m= m^{(u)}_1}^{M} \frac{K p^{pop}_{m} }{a(n)\hat{X}_{m}}.
\end{align*}
Substituting for $\hat{X}_{m}$ using \eqref{eq:oprep_s} and for $p^{pop}_{m}$ using~\eqref{eq:zipf}, we obtain
\begin{align}\label{eq:mdav_s}
D_{avg}(\!n\!)\!=\!\frac{KH_\alpha(m^{(u)}_1\!-1)}{H_\alpha(M)}\!+\!\frac{K\!\left(\!H_\frac{\alpha}{2}(M)\!-H_\frac{\alpha}{2}(m^{(u)}_1\!-1\!) \right)^2 }{H_\alpha(M)\left(\! na(n)-m^{(u)}_1\!\right)},
\end{align}
where $H_\alpha(M)$ and $H_\frac{\alpha}{2}(M)$ are given in \eqref{eq:H}  and \eqref{eq:H2}, respectively. When $m^{(u)}_1= \Theta(M)$, the $D_{avg}(n)$ is expressed only as the first term on the RHS of \eqref{eq:mdav_s}, which thus results in $D_{avg}(n)=\Theta(K)$. On the other hand, when $m^{(u)}_1= o(M)$, the first and second terms on the RHS of \eqref{eq:mdav_s} scale as $O(K)$ and $\Theta\left(\frac{K\left(H_\frac{\alpha}{2}(M)\right)^2}{H_\alpha(M) na(n)}\right)$, respectively. Hence, \eqref{eq:delay_s2} holds, which completes the proof of Theorem \ref{th:delay_s}.

%%%%%%%%%%%%%%%%%%%%%%%%%%%%%%%%%%%%%%%%%%%%%%%%%%%%%%%%%%%%%%%%%%%%%%%%%%%%%%%%%%%%%%%%%%%%%%%%%%%%%%%%%%%%%%%%%%%%%%%%%%%%%%%%%%%%%%%%%
%%%%%%%%%%%%%%%%%%%%%%%%%%%%%%%%%%%%%%%%%%%%%%%%%%%%%%%%%%%%%%%%%%%%%%%%%%%%%%%%%%%%%%%%%%%%%%%%%%%%%%%%%%%%%%%%%%%%%%%%%%%%%%%%%%%%%%%%%
%%%%%%%%%%%%%%%%%%%%%%%%%%%%%%%%%%%%%%%%%%%%%%%%%%%%%%%%%%%%%%%%%Proofs of MDS-Coded Section%%%%%%%%%%%%%%%%%%%%%%%%%%%%%%%%%%%%%%%%%%%%%
%%%%%%%%%%%%%%%%%%%%%%%%%%%%%%%%%%%%%%%%%%%%%%%%%%%%%%%%%%%%%%%%%%%%%%%%%%%%%%%%%%%%%%%%%%%%%%%%%%%%%%%%%%%%%%%%%%%%%%%%%%%%%%%%%%%%%%%%%
%%%%%%%%%%%%%%%%%%%%%%%%%%%%%%%%%%%%%%%%%%%%%%%%%%%%%%%%%%%%%%%%%%%%%%%%%%%%%%%%%%%%%%%%%%%%%%%%%%%%%%%%%%%%%%%%%%%%%%%%%%%%%%%%%%%%%%%%%

%%\section{Proof of Lemma \ref{le:1_c}} \label{AppendixA_c}
\subsection{Proof of Lemma \ref{le:1_c}} \label{AppendixA_c}
We begin our proof by showing that the optimal cache allocation strategy is non-increasing for each regime. First, note that the transfer delay for the content associated with Regime I\(^{(c)}\) is the same in order sense for any value of $r_m =\Omega (a(n)^{-1})$ as shown in \eqref{eq:davgc1}. Thus, we choose $\hat{r}_m= a(n)^{-1}$ for $m \in \mathcal{I}^{(c)}_1$, which is non-increasing. Now, let us focus on Regime II\(^{(c)}\). From the stationary conditions in \eqref{eq:KKT1_c}, we have 
\begin{equation} \label{eq:KKT1r2_c} 
 - \frac{p^{pop}_{m}K}{a(n) \hat{r}^{2}_{m}}+\hat{\delta}-\hat{\sigma}_m +\hat{\mu}_m =0  \hspace{1cm} m \in \mathcal{I}^{(c)}_2.
\end{equation} 
\noindent From conditions~\eqref{eq:KKT2_c} and ~\eqref{eq:KKT6_c}, we have $\hat{\sigma}_m=0$ and $\hat{\mu}_m=0$ for $m \in \mathcal{I}^{(c)}_{2}$. Using \eqref{eq:KKT1r2_c}, we have
\begin{align} \label{eq:opxr2_c} 
\hat{r}_{m} = \frac{\sqrt{p^{pop}_m}}{\sum_{\widetilde{m} \in \mathcal{I}^{(c)}_{2}} \sqrt{p^{pop}_{\widetilde{m}} }} \sum_{\widetilde{m} \in \mathcal{I}^{(c)}_{2} } \hat{r}_{\widetilde{m}}, \hspace{1cm} m \in \mathcal{I}^{(c)}_{2}.
\end{align}
Hence, the optimal cache allocation strategy $\left\{\hat{r}_{m}\right\}_{m \in \mathcal{I}^{(c)}_{2}}$ is non-increasing in Regime II\(^{(c)}\). 
Similarly as in the case of Regime I\(^{(c)}\), the transfer delay for the content associated with Regime III\(^{(c)}\) is the same in order sense for any value of $r_m =o(K^{1+\epsilon})$ and $r_m =\Omega(K)$ as shown in \eqref{eq:davgc3}. Thus, we choose $\hat{r}_m= K$ for $m \in \mathcal{I}^{(c)}_3$, which is non-increasing. 
Now, we are ready to finalize the proof of Lemma \ref{le:1_c}.  Consider any content object $j \in \mathcal{I}^{(c)}_1$. Then, using~\eqref{eq:KKT1_c},~\eqref{eq:dm_r1}, and the fact that $\hat{\sigma}_m = 0$ and $\hat{\mu}_m = 0$ for $m \in \mathcal{I}^{(c)}_1\bigcup \mathcal{I}^{(c)}_2$, we have $\hat{\delta} =  \frac{K p^{pop}_{m^{(c)}_{1}}}{a(n) \hat{r}^{2}_{m^{(c)}_{1}}}= p^{pop}_{j}\left(1- \frac{1}{ a(n) (\hat{r}_{j}-K) }\right)  > 0$. Since  $\hat{r}_{j}= \Theta\left(a(n)^{-1}\right)$ and $\hat{r}_{m^{(c)}_{1}}= o\left(a(n)^{-1}\right)$, we obtain $p_j > p_{m^{(c)}_{1}}$, thus resulting in $j < m^{(c)}_{1}$. 
Now, consider any content object $l \in \mathcal{I}^{(c)}_2$. Then, using~\eqref{eq:KKT1_c},~\eqref{eq:dm_r3}, and the fact that $\hat{\mu}_m = 0$ for $m \in \mathcal{M}$ and $\hat{\sigma}_m = 0$ for $m \in \mathcal{I}^{(c)}_2$, we have $\hat{\delta} = \frac{K p^{pop}_{l}}{a(n) \hat{r}^{2}_{l}}  = \frac{K p^{pop}_{m^{(c)}_{2}}}{a(n) \hat{r}_{m^{(c)}_{2}}\left(\hat{r}_{m^{(c)}_{2}}-K\right)} + \hat{\sigma}_{m^{(c)}_{2}} >0 $. Since $\hat{r}_{m^{(c)}_{2}}= \Theta(K) $, $\hat{r}_{l}= \Omega\left( K^{1+\epsilon}\right)$, and $\hat{\sigma}_{m^{(c)}_{2}} \geq 0$, we obtain $p^{pop}_l > p^{pop}_{m^{(c)}_{2}}$, thus resulting in $l < m^{(c)}_{2}$. Hence, we finally have $ j < m^{(c)}_{1} < m^{(c)}_{2}$. This completes the proof of the Lemma \ref{le:1_c}.
\subsection{Proof of Proposition \ref{th:oprep_c}}\label{AppendixD_c}
Let us first characterize the content indice $m^{(c)}_{1}$ and $m^{(c)}_{2}$, each of which identifies the boundary between regimes. Since $m^{(c)}_2-1 \in  \mathcal{I}^{(c)}_{2}$ is the largest index such that $\hat{r}_{m^{(c)}_{2}-1}= \Omega\left( K^{1+\epsilon}\right)$, using \eqref{eq:opxr2_c} yields
\begin{align} \label{eq:m21_c}
(m^{(c)}_{2})^{-\frac{\alpha}{2}}  = \Theta \left( K^{1+\epsilon} \frac{\sum_{\widetilde{m} = m^{(c)}_{1}}^{m^{(c)}_{2}-1} \widetilde{m}^{-\frac{\alpha}{2}}}{\sum_{\widetilde{m}= m^{(c)}_{1} }^{m^{(c)}_{2}-1} \hat{r}_{\widetilde{m}}} \right).
\end{align}
Here, it follows that $\sum_{\widetilde{m} = m^{(c)}_{1} }^{m^{(c)}_{2}-1} \hat{r}_{\widetilde{m}}= Sn - (m^{(c)}_{1}-1)a(n)^{-1}- (M-m^{(c)}_{2}+1)K $ such that the condition \eqref{eq:KKT3_c} is fulfilled. We thus obtain 
\begin{align}\label{eq:m22_c}
m^{(c)}_{2}=\left(\frac{Sn-(m^{(c)}_{1}-1) a(n)^{-1}- (M-m^{(c)}_{2}+1)K}{K^{1+\epsilon}H_{\frac{\alpha}{2}}(m^{(c)}_{2})} \right)^{\frac{2}{\alpha}}.  
\end{align}
As $m^{(c)}_{1} \in \mathcal{I}^{(c)}_{2}$ is the smallest index such that $\hat{r}_{m^{(c)}_{1}} = o\left(a(n)^{-1}\right)$, using~\eqref{eq:opxr2_c} yields
\begin{align*} 
\left(m^{(c)}_{1} -1\right)^{-\frac{\alpha}{2}}  = \Theta \left( a(n)^{-1} \frac{\sum_{\widetilde{m} = m^{(c)}_{1}}^{m^{(c)}_{2}-1} \widetilde{m}^{-\frac{\alpha}{2}}}{\sum_{\widetilde{m}= m^{(c)}_{1} }^{m^{(c)}_{2}-1} \hat{r}_{\widetilde{m}}} \right).
\end{align*}
\noindent From \eqref{eq:m21_c}, we obtain
\begin{align}\label{eq:m12_c}
m^{(c)}_{1}  = \Theta \left( \left(\frac{K^{1+\epsilon}}{a(n)^{-1}}\right) ^{\frac{2}{\alpha}} m^{(c)}_{2} \right).
\end{align}
Using the fact that $S=\Theta(K)$ and combining \eqref{eq:m22_c} and \eqref{eq:m12_c}, we have
\begin{align}
m^{(c)}_{2}\!\! = \left(\frac{\!\!n - m^{(c)}_{2}\left(K^{1+\epsilon}a(n)\right)^{\frac{2}{\alpha}-1}- (M-m^{(c)}_{2}+1) }{H_{\frac{\alpha}{2}}(m^{(c)}_{2})}\right)^{\frac{2}{\alpha}}.
\end{align}
From Lemma \ref{le:1_c}, the optimized cache allocation strategy is given in \eqref{eq:oprep_c} since $\sum_{\widetilde{m} = m^{(c)}_{1} }^{m^{(c)}_{2}-1} \hat{r}_{\widetilde{m}}= Sn - (m^{(c)}_{1}-1)a(n)^{-1}- (M-m^{(c)}_{2}+1)K $. This completes the proof for Proposition \ref{th:oprep_c}. 
\subsection{Proof of Theorem \ref{th:delay_c}}\label{AppendixE_c}
In Regimes I$^{(c)}$, II$^{(c)}$, and III$^{(c)}$, the minimum $D_{avg}(n)$ resulting from the optimal cache allocation strategy $\left\{\hat{r}_{m}\right\}_{m=1}^M$ in \eqref{eq:oprep_c} is given by
\begin{equation*}  
D_{avg}(n)=\!\!\!\sum_{m=1}^{m^{(c)}_{1}-1}\!\!Kp^{pop}_{m}+\!\!\!\!\sum_{m=m^{(c)}_{1}}^{m^{(c)}_{2}-1}\!\!\!\frac{p^{pop}_{m}K}{a(n)\hat{r}_m}\!+\!\!\!\!\sum_{m = m^{(c)}_{2}}^{M}\frac{p^{pop}_{m}}{a(n)}\log K.
\end{equation*} 
\noindent Substituting for $\hat{r}_{m}$ using \eqref{eq:oprep_c} and for $p^{pop}_{m}$ using~\eqref{eq:zipf}, we obtain
\begin{align}\label{eq:mdav_c}
D_{avg}(n) &=  \frac{K H_{\alpha}(m^{(c)}_{1}-1)}{H_{\alpha}(M)} + \frac{K (H_{\frac{\alpha}{2}}(m^{(c)}_{2}))^2}{H_{\alpha}(M)a(n) S^{(c)}} \nonumber \\&  + \frac{\log K\left( H_{\alpha}(M)- H_{\alpha}(m^{(c)}_{2})\right)}{a(n) H_{\alpha}(M)}.
\end{align}
\noindent where $H_\alpha(M)$ and $H_\frac{\alpha}{2}(M)$ are given in \eqref{eq:H} and \eqref{eq:H2}, respectively. When $m^{(c)}_1= \Theta(M)$, the $D_{avg}(n)$ is expressed only as the first term on the RHS of \eqref{eq:mdav_c}, which thus results in $D_{avg}(n)=\Theta(K)$. When $m^{(c)}_1=o(M)$ and $m^{(c)}_2= \Theta(M)$, the $D_{avg}(n)$ is expressed only as the first and second terms on the RHS of \eqref{eq:mdav_c}, which thus results in $D_{avg}(n)= \Theta\left(\max \left\{K,\frac{\left(H_\frac{\alpha}{2}(M)\right)^2}{H_\alpha(M) na(n)}\right\}\right)$. On the other hand, when $m^{(c)}_2= o(M)$, the $D_{avg}(n)$ scales as $\Theta\!\left(\!\!\max\!\left\{\!\!K,\frac{a(n)^{-1}\!\!\left(\!H_\frac{\alpha}{2}(\!m^{(c)}_2\!)\!\right)^2}{H_\alpha(M) (n-M)}, \frac{\log K}{a(n)}\!\!\right\}\!\!\right)$. This completes the proof of Theorem \ref{th:delay_c}.

\ifCLASSOPTIONcaptionsoff
  \newpage
\fi
% references section


\begin{thebibliography}{10}
\providecommand{\url}[1]{#1}
\def\UrlFont{\rmfamily}
\providecommand{\newblock}{\relax}
\providecommand{\bibinfo}[2]{#2}
\providecommand\BIBentrySTDinterwordspacing{\spaceskip=0pt\relax}
\providecommand\BIBentryALTinterwordstretchfactor{4}
\providecommand\BIBentryALTinterwordspacing{\spaceskip=\fontdimen2\font
plus \BIBentryALTinterwordstretchfactor\fontdimen3\font minus
  \fontdimen4\font\relax}
\providecommand\BIBforeignlanguage[2]{{%
\expandafter\ifx\csname l@#1\endcsname\relax
\typeout{** WARNING: IEEEtran.bst: No hyphenation pattern has been}%
\typeout{** loaded for the language `#1'. Using the pattern for}%
\typeout{** the default language instead.}%
\else \language=\csname l@#1\endcsname \fi #2}}

\bibitem{R1-2} V. Jacobson, D. K. Smetters, J. D. Thornton, M. F. Plass, N. H. Briggs, and R. L. Braynard, ``Networking named content," \emph{Commun. ACM}, vol. 55, no. 1, pp. 117--124, Jan. 2012. %1

\bibitem{R1-3}M. Ji, G. Caire, and A. F. Molisch, ``Wireless device-to-device caching networks: Basic principles and system performance," \emph{IEEE J. Sel. Areas Commun.}, vol. 34, no. 1, pp. 176--189, Jan. 2016. %2

\bibitem{gupta} P. Gupta and P. R. Kumar, ``The capacity of wireless networks," \emph{IEEE Trans. Inf. Theory}, vol. 46, no. 2, pp. 388--404, Mar. 2000. %3

\bibitem{R1-5} M. Franceschetti, O. Dousse, D. N. C. Tse, and P. Thiran, ``Closing the gap in the capacity of wireless networks via percolation theory," \emph{IEEE Trans. Inf. Theory}, vol. 53, no. 3, pp. 1009--1018, Mar. 2007. %4

\bibitem{R1-6} P. Gupta and P. R. Kumar, ``Towards an information theory of large networks: An achievable rate region," \emph{IEEE Trans. Inf. Theory}, vol. 49, no. 8, pp. 1877--1894, Aug. 2003. %5

\bibitem{R1-8} W.-Y. Shin, S.-Y. Chung, and Y. H. Lee, ``Parallel opportunistic routing in wireless networks," \emph{IEEE Trans. Inf. Theory}, vol. 59, no. 10, pp. 6290--6300, Oct. 2013. %7

\bibitem{grossglauser} M. Grossglauser and D. N. C. Tse, ``Mobility increases the capacity of ad hoc wireless networks," \emph{IEEE/ACM Trans. Netw.}, vol. 10, no. 4, pp. 477--486, Aug. 2002. %8

\bibitem{algamal}A. El Gamal, J. Mammen, B. Prabhakar, and D. Shah, ``Optimal throughput--delay scaling in wireless networks--Part I: The fluid model," \emph{IEEE Trans. Inf. Theory}, vol. 52, no. 6, pp. 2568--2592, Jun. 2006. %9

\bibitem{R1-9} A. \"{O}zg\"{u}r, O. L\'{e}v\^{e}que, and D. N. C. Tse, ``Hierarchical cooperation achieves optimal capacity scaling in ad hoc networks," \emph{IEEE Trans. Inf. Theory}, vol. 53, no. 10, pp. 3549--3572, Oct. 2007. %10

\bibitem{R1-16} B. Liu, Z. Liu, and D. Towsley, ``On the capacity of hybrid wireless networks," in \emph{Proc. IEEE INFOCOM}, San Francisco, CA, Mar./Apr. 2003, pp. 1543--1552. %11

\bibitem{R1-17} W.-Y. Shin, S.-W. Jeon, N. Devroye, M. H. Vu, S.-Y. Chung, Y. H. Lee, and V. Tarokh,  ``Improved capacity scaling in wireless networks with infrastructure," \emph{IEEE Trans. Inf. Theory}, vol. 57, no. 8, pp. 5088--5102, Aug. 2011. %12

\bibitem{R1-13} G. Zhang, Y. Xu, X. Wang, and M. Guizani, ``Capacity of hybrid wireless networks with directional antennas and delay constraint," \emph{IEEE Trans. Commun.}, vol. 58, no. 7, pp. 2097--2106, July 2010. %13

\bibitem{R1-15} J. Yoon, W.-Y. Shin, and S.-W. Jeon, ``Elastic routing in ad hoc networks with directional antennas," {\em IEEE Trans. Mobile Comput.}, vol. 16, no. 12, pp. 3334--3346, Dec. 2017. %14

\bibitem{alfano}G. Alfano, M. Garetto, and E. Leonardi, ``Content-centric wireless networks with limited buffers: When mobility hurts," \emph{IEEE/ACM Trans. Netw.}, vol. 24, no. 1, pp. 299--311, Feb. 2016. %15

\bibitem{as_law}S. Gitzenis, G. S. Paschos, and L. Tassiulas, ``Asymptotic laws for joint content replication and delivery in wireless networks," \emph{IEEE Trans. Inf. Theory}, vol. 59, no. 5, pp. 2760--2776, May. 2013. %16

\bibitem{jeon}S.-W. Jeon, S.-N. Hong, M. Ji, G. Caire, and A. F. Molisch ``Wireless multihop device-to-device caching networks," {\em IEEE Trans. Inf. Theory}, vol. 63, no. 3, pp. 1662--1676, Mar. 2017. %17

\bibitem{R1-19}M. Ji, G. Caire, and A. F. Molisch, ``The throughput--outage tradeoff of wireless one-hop caching networks," \emph{IEEE Trans. Inf. Theory}, vol. 61, no. 12, pp. 6833--6859, Dec. 2015. %18

\bibitem{liu}X. Liu, K. Zheng, J. Zhao, X. Y. Liu, X. Wang, and X. Di, ``Information-centric networks with correlated mobility," \emph{IEEE Trans. Veh. Technol.}, vol. 66, no. 5, pp. 4256--4270, May. 2017. %19

\bibitem{anh}T.-A. Do, S.-W. Jeon, and W.-Y. Shin, ``Caching in mobile HetNets: A throughput-delay trade-off perspective," in \emph{Proc. IEEE Int. Symp. Inf. Theory (ISIT)}, Barcelona, Spain, Jul. 2016, pp. 1247-1251. %20

\bibitem{Maddah1}M. A. Maddah-Ali and U. Niesen, ``Fundamental limits of caching," \emph{IEEE Trans. Inf. Theory}, vol. 60, no. 5, pp. 2856--2867, May. 2014. %21

\bibitem{Maddah2}M. A. Maddah-Ali and U. Niesen, ``Decentralized coded caching attains order-optimal memory-rate tradeoff," \emph{IEEE/ACM Trans. Netw.}, vol. 23, no. 4, pp. 1029--1040, Aug. 2014. 

\bibitem{lim} S. H. Lim, C.-Y. Wang, and M. C. Gastpar, ``Information-theoretic caching: The multi-user case," \emph{IEEE Trans. Inf. Theory}, vol. 63, no. 11, pp. 7018--7037, Jul. 2017. %1

\bibitem{d2d} M. Ji, G. Caire and A. F. Molisch, ``Fundamental limits of caching in wireless D2D networks," \emph{IEEE Trans. Inf. Theory}, vol. 62, no. 2, pp. 849--869, Feb. 2016. %36

\bibitem{mds1} V. Bioglio, F. Gabry, and I. Land, ``Optimizing MDS codes for caching at the edge," in \emph{Proc. IEEE GLOBECOM}, San Diego, CA, Dec. 2015, pp. 1--6. %24

\bibitem{mds2}J. Pedersen, A. Graell i Amat, I. Andriyanova, and F. Br{\"{a}}nnstr{\"{o}}m, ``Optimizing MDS coded caching in wireless networks with device-to-device communication," preprint, [Online]. Available: http://arxiv.org/abs/1701.06289. %25

\bibitem{femto} K. Shanmugam, N. Golrezaei, A. G. Dimakis, A. F. Molisch and G. Caire, ``Femtocaching: Wireless content delivery through distributed caching helpers," \emph{IEEE Trans. Inf. Theory}, vol. 59, no. 12, pp. 8402--8413, Dec. 2013. %35

\bibitem{bigo}D. E. Knuth, ``Big Omicron and big Omega and big Theta," \emph{ACM SIGACT News}, vol. 8, no. 2, pp. 18--24, Apr.–-Jun. 1976. %27

\bibitem{gf}N. Jacobson, \emph{Lectures in Abstract Algebra: III. Theory of Fields and Galois Theory}. Springer Science \& Business Media, 2012, vol. 32.  %28

\bibitem{mds}I. Tamo, Z. Wang, and J. Bruck, ``MDS array codes with optimal rebuilding," in \emph{Proc. IEEE Int. Symp. Inf. Theory (ISIT)}, St. Petersburg, Russia, Jul. 2011, pp. 1240-1244. %29

\bibitem{milad}M. Mahdian and E.M. Yeh,, ``Throughput and delay scaling of content-centric ad hoc and heterogeneous wireless networks," \emph{IEEE/ACM Trans. Netw.}, vol. 25, no. 5, pp. 3030--3043, Oct. 2017. %30

\bibitem{zipf}C. Fricker, P. Robert, J. Roberts, and N. Sbihi, ``Impact of traffic mix on caching performance in a content-centric network," in \emph{Proc. IEEE INFOCOM Workshop on Emerging Choices in Named-Oriented Netw. (NoMEN)}, Orlando, FL, Mar. 2012, pp. 310--315. %31
\bibitem{buffer} T.-Y. Huang, R. Johari, N. McKeown, M. Trunnell, and M. Watson, ``A buffer-based approach to rate adaptation: Evidence from a large video streaming service," \emph{SIGCOMM Comput. Commun. Rev.}, vol. 44, no. 4, pp. 187--198, Oct. 2014. %34
\bibitem{renewal}P.V. Mieghem, \emph{Performance Analysis of Communications Networks and Systems}. New York, NY, USA: Cambridge Univ. Press, 2005. %32

\bibitem{rw} L. Ying, S. Yang, and R. Srikant, ``Optimal delay--throughput tradeoffs in mobile ad hoc networks," \emph{IEEE Trans. Inf. Theory}, vol. 54, no. 9, pp. 4119--4143, Sept. 2008. %33
\end{thebibliography}
\end{document}